\documentclass[onecolumn,draftclsnofoot]{IEEEtran}

\pdfoutput=1

\usepackage{cite}
\usepackage{amsmath, amsthm, amsfonts, amssymb, multirow}

\usepackage{amsthm,amsfonts,amsmath,amssymb,dsfont,color}
\usepackage{commath,enumerate}
\usepackage{mathtools}
\usepackage[ruled,vlined, linesnumbered]{algorithm2e}
\usepackage{cite}
\usepackage{hyperref}

\usepackage{hyperref}

\usepackage{caption, subcaption, multirow} 
\usepackage{makecell}

\newcommand{\define}{\triangleq}
\newcommand{\alg}{Directional Discrepancy }

\theoremstyle{theorem}
\newtheorem{thm}{{Theorem}}[]
\newtheorem*{thm*}{{Theorem}}

\newtheorem{lem}{Lemma}

\theoremstyle{remark}
\newtheorem{rem}{Remark}
\theoremstyle{definition}
\newtheorem{defn}{Definition}
\newtheorem{conj}{Conjecture}
\def\F{\mathcal{F}}

\newcommand{\mvec}[1]{{\boldsymbol #1}}
\newcommand{\normalvec}[1]{\Vec{\boldsymbol #1}}

\usepackage{xcolor}

\usepackage{authblk}
\begin{document}

\title{A practical algorithm to calculate Cap Discrepancy}
\author[1]{Milad Bakhshizadeh}
\affil[1]{Columbia University \\ {mb4041@columbia.edu}}

\author[2]{Ali Kamalinejad}
\affil[2]{University Tehran \\ {kamalinejad.a@ut.ac.ir}}

\author[3]{Mina Latifi}
\affil[3]{Alzahra University \\ {minlat1375@gmail.com}}

% make the title area
\maketitle 
% As a general rule, do not put math, special symbols or citations
% in the abstract
\begin{abstract}
Uniform distribution of the points has been of interest to researchers for a long time and has applications in different areas of Mathematics and Computer Science.  One of the well-known measures to evaluate the uniformity of a given distribution is Discrepancy, which assesses the difference between the Uniform distribution and the empirical distribution given by putting mass points at the points of the given set.  While Discrepancy is very useful to measure uniformity, it is computationally challenging to be calculated accurately.  
We introduce the concept of directed Discrepancy based on which we have developed an algorithm, called Directional Discrepancy, that can offer accurate approximation for the cap Discrepancy of a finite set distributed on the unit Sphere, $\mathds{S}^2.$  We also analyze the time complexity of the Directional Discrepancy algorithm precisely; and practically evaluate its capacity by calculating the Cap Discrepancy of a specific distribution, Polar Coordinates, which aims to distribute points uniformly on the Sphere.  
\end{abstract}

\section{Introduction}

\subsection{Uniform distribution on the Sphere}
The problem of uniformly distributing a finite set of points  on the Sphere has a long history and has applications in many areas of mathematics like Combinatorics and Numerical Analysis as well as in Finance, Computer Graphics, and Computational Physics \cite{Matousek}.
Defining a suitable measure of uniformity and finding optimal distribution with respect to a given measure are both challenging problems. Hence, there have been extensive efforts to address these problems which yield in several metrics to measure the distance from the Uniform distribution and many good structures to locate the points on the sphere  \cite{katanforoush2003, lubotzky1987hecke, aistleitner2012point}.  At the present work, we will focus on a specific measure of uniformity, namely Discrepancy, which evaluates how well the empirical mass points at the given set approximates the probability distribution on the surface $\mathds{S}^2$ induced by the Lebesgue measure.  We review the concept of Discrepancy in detail in the next section.

\subsection{Discrepancy}

This section aims to review the concept of Discrepancy and some of its applications.  Let $\mu$ be any finite measure on $\mathds{S}^2$, and $\mathcal{F}$ denotes a family of measurable subsets of the Sphere.  If the  distribution of a point set is relatively close to $\mu$ we expect that mass points at this set offer a good approximation of $\mu$, possibly after an appropriate rescaling.  Suppose $P$ denotes a point set of size $t$ on the Sphere.  The difference
$\envert{ \frac{\envert{C \cap P}}{t} - \frac{\mu(C)}{\mu(\mathds{S}^2)}}, \quad C \in \mathcal{F}$ shows how well this approximation is for the subset $C$.  This intuition leads one to the following measure of uniformity, called {Discrepancy}. 

\begin{defn} \label{def:Discrepancy}

Let $\F$ be a family of measurable sets, $\mu$ be a measure, and $P$ be a 
set of $t$ points on the Sphere.  The Discrepancy of $P$ with respect to $C \in \F$ is defined

\begin{equation} \label{eq:Discrepancy definition wrt one set}
 {\rm Dis}_C(P) \define \envert{ \frac{\envert{C \cap P}}{t} - \frac{\mu(C)}{\mu(\mathds{S}^2)}},
\end{equation}
and Discrepancy with respect to the family $\F$ is defined as
\begin{equation} \label{eq:Discrepancy wrt to famlily}
{\rm Dis}_{\F}(P) \define \sup_{C \in \mathcal{F}} Dis_C(P).
\end{equation}
\end{defn}

One of the most common choices for $\F$ is the family of caps (subsets of the Sphere cut by half spaces).  Also, the Spherical measure is a natural choice for $\mu$.  When we are in this setting we denote ${\rm Dis}_{\F}$ by a simpler notation ${\rm Dis}$. For the rest of the paper, $\F$ denotes the caps family and $\mu$ is the Spherical measure.

\begin{figure}[h]
    \centering
    \includegraphics{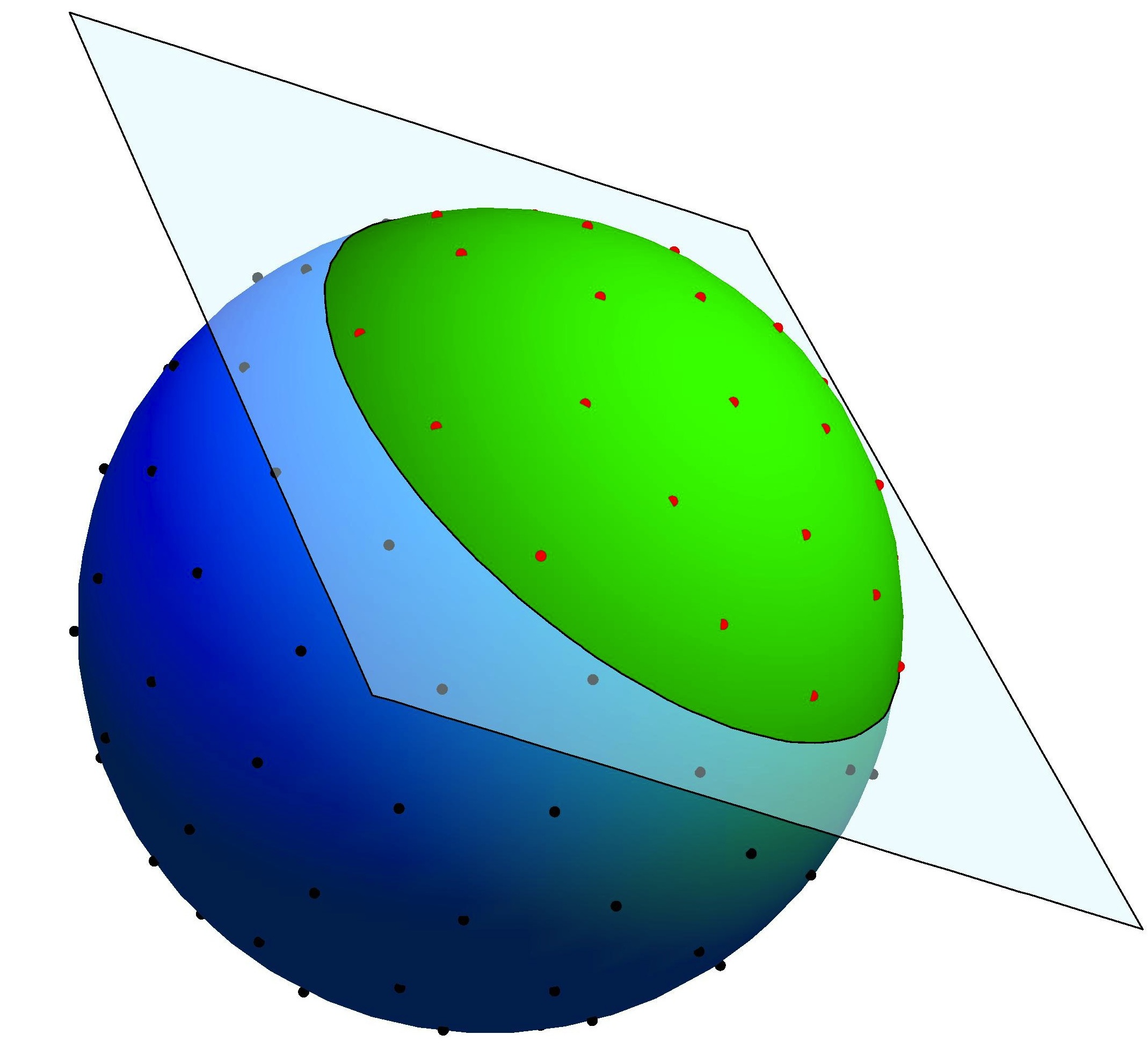}
    \caption{Discrepancy of a point set with respect to a cap}
    \label{fig:cap discrepancy}
\end{figure}

The study of point sets with low Discrepancy, i.e. almost uniformly distributed, has applications in a variety of fields of study including, but not limited to, Computer Science \cite{matouvsek1996derandomization}, Image Processing 
\cite{nair1999image}, Ramsey Theory \cite{t1983irregularities}, Number Theory \cite{beck1991flat}, and Numerical Analysis. In particular, the Koksma–Hlawka inequality states how one can bound the error of numerical integration of a function in terms of the Discrepancy of evaluation points and the total variation of the integrant \cite{Matousek}.

\begin{thm*}[Koksma–Hlawka Inequality]
If $f : [0, 1]^d \to \mathds{R}$ and $P \subset [0, 1]^d$ be a set with $t$ points, we have
\begin{equation*} 
    \abs{ \int_{\mathds{R}} f(x) dx - \frac{1}{t} \sum_{p \in P} f(p) } \leq \frac{1}{t} Dis(P) V(f),
\end{equation*}
where $Dis(P)$ denotes the Discrepancy of the set $P$ which is defined in a similar way for $[0, 1]^d$, and $V(f)$ is the total variation of $f$.
\end{thm*}

\subsection{Related works}

Discrepancy and structures of finite points with small Discrepancies have been appeared in different researches.  \cite{beck-1984} proves, for $\mathds{S}^d \subset \mathds{R}^{d + 1}$, when the size of a point set is set to be $t$, the Discrepancy is lower bounded by $c t^{- \frac{1}{2} - \frac{1}{2 d}}$, regardless of the distribution of the points,  where $c$ is a universal constant.  In the particular case of two-dimensional sphere, \cite{beck-1984} states that Discrepancy of any point set with size $t$ is $O(t^{-\frac{3}{4}})$.  There are random structures whose Discrepancy is off from this lower bound only by a logarithmic factor \cite{Matousek}.
However, deterministic distributions have been seen to have higher Discrepancy.  \cite{lubotzky1987hecke} introduces a structure based on the action of the symmetric group on the Sphere with cap Discrepancy of order $t^{-\frac{1}{3}}$.  \cite{aistleitner2012point} constructs a structure by using the Fibonacci sequence which can reach the Discrepancy of order $t^{- \frac{1}{2}}$.  \cite{katanforoush2003} investigates other measures of uniformity based on the Coulomb Potential metric for some point structures on the Sphere.  It also introduces a structure, called Polar Coordinates, and shows this structure is pretty uniform with respect to the measures they studied.  This result inspired us to examine Polar Coordinates with respect to Discrepancy as well and it turns out one can get low Discrepancy configuration, by slightly modifying Polar Coordinates.  Details of this structure and our investigations will be discussed in section \ref{sec:an application}.
Due to the best of our knowledge no deterministic structure of points has been known to have Discrepancy lower than $O(t^{- \frac{1}{2}})$ which seems to be the order of the Discrepancy for both Polar Coordinates and Fibonacci distribution introduced at \cite{aistleitner2012point}.

\subsection{Summary of our contribution}

The Discrepancy is an intuitive measure in the sense that it assigns large values to the point sets who have a high concentration of points on a cap or leave a sparse cap on the sphere. Obtaining the exact value of the Discrepancy for a finite set, however, is computationally challenging. This is because one needs to take supremum over an infinite family of caps to obtain the Discrepancy.  It is not hard to check one only needs to check finitely many caps, those whose boundary passes through a few points from the given set.  Nevertheless, the number of candidate caps grows like $t^3$, where $t$ as above is the size of the point set.  Given calculating the Discrepancy for a fixed cap has a time complexity of order $O(t)$, the naive method of checking all  candidate caps will run in $O(t^4)$, which is a prohibitively large time complexity.  

We defined the notion of {directed Discrepancy} based on which we proposed an efficient algorithm, called {Directional Discrepancy}, to approximate the Discrepancy of a given point set.  The Time complexity of the Directional Discrepancy algorithm depends on how uniformly the point set is distributed as well as the level of acceptable error for the approximation.  For the point sets that are expected to have low Discrepancy, the Directional Discrepancy algorithm can run in linear time complexity, i.e. $O(t)$, to find an approximation with a fixed additive error.  If one wishes to obtain small relative errors for Discrepancy, the algorithm's run time growth is upper bounded by $O(t^{\frac{5}{2}})$.  Moreover, for the known deterministic distributions, whose Discrepancy is of order $O(t^{\frac{-1}{2}})$, we proved the \alg has time complexity upper bounded by $O(t^2)$ to offer an approximation with a relatively  small deviation from the Discrepancy.  Apart from reducing the time complexity of the calculation of Discrepancy, our proposed algorithm can offer informative data for a given distribution, like the direction at which the maximum Discrepancy occurs and local behavior of directed Discrepancy all over the Sphere.

To evaluate the performance of the Directional Discrepancy algorithm practically we applied it to calculate the Discrepancy of a specific distribution called Polar Coordinates, which was introduced in \cite{katanforoush2003}.   It turns out slightly modification of Polar Coordinates yields one of the best known deterministic structures to minimize Discrepancy.  Moreover, our experimental investigations lead us to a conjecture, for which we have enough experimental evidence, that reveals an important structural feature of Polar Coordinates.  If we assume the conjecture holds, we can provably upper bound the Discrepancy of the Polar Coordinates with $O(t^{-\frac{1}{2}})$ which is the smallest known order of Discrepancy for a deterministic distribution.

We believe Directional Discrepancy can be used in a similar way to offer informative data about other finite structures as well.
The implementation of the proposed approximation and the run results for Polar Coordinates can be found at: \href{https://github.com/minalatifi/DirectionalDiscrepancy/tree/main/Results}{https://github.com/minalatifi/DirectionalDiscrepancy/tree/main/Results}.
Due to the best of our knowledge, Directional Discrepancy is the first algorithm proposed to approximate cap Discrepancy in a timely manner.

\subsection{Organization of the paper}
The rest of the paper is organized as follows. Section \ref{sec: directional Discrepancy} introduces the directed Discrepancy as well as the \alg algorithm, and discusses the theoretical properties of these two.  Section \ref{sec:an application} reports the results we obtained by using the \alg algorithm on Polar Coordinates distribution which is followed by the conclusion section \ref{sec: conclusion}.

\section{Directional Discrepancy Algorithm} \label{sec: directional Discrepancy}

In this section, we introduce the notion of the directed Discrepancy as well as \alg algorithm.  Moreover, we state and prove all the requirements for the proof of the correctness of the \alg algorithm. 

\subsection{Background}

Let us start with the background needed to introduce the \alg algorithm.  In the following, we define directed Discrepancy and study its properties.

\begin{defn} \label{def:Directed Discrepancy}
Given a normal vector $\normalvec{v}$, {directed Discrepancy} in direction $\normalvec{v}$ is Discrepancy of the subfamily whose caps' normal vector is $\normalvec{v}$.  In other words, if we let
$\mathcal{F}_{\normalvec{v}} = \cbr{C \in \mathcal{F} \sVert[2] C = \cbr{\mvec{x} \in \mathds{R}^3  \sVert[2] \langle \mvec{x},\normalvec{v} \rangle \geq h}, h \in \mathds{R}} $,
then directed Discrepancy is ${\rm Dis}_{\F_{\normalvec{v}}}(A)$, defined in the equation \eqref{eq:Discrepancy wrt to famlily}.  For simplicity of the notation, we will use ${\rm Dis}_{\normalvec{v}}$ for directional Discrepancy in the direction of $\normalvec{v}$.  Note that a cap can be represented by the plane that cuts it.  Hence, we may interchangeably denote a cap or its boundary plane throughout the paper.
\end{defn}

\begin{figure}[h!]
\begin{subfigure}{0.45 \textwidth}
\includegraphics[width = .8 \textwidth]{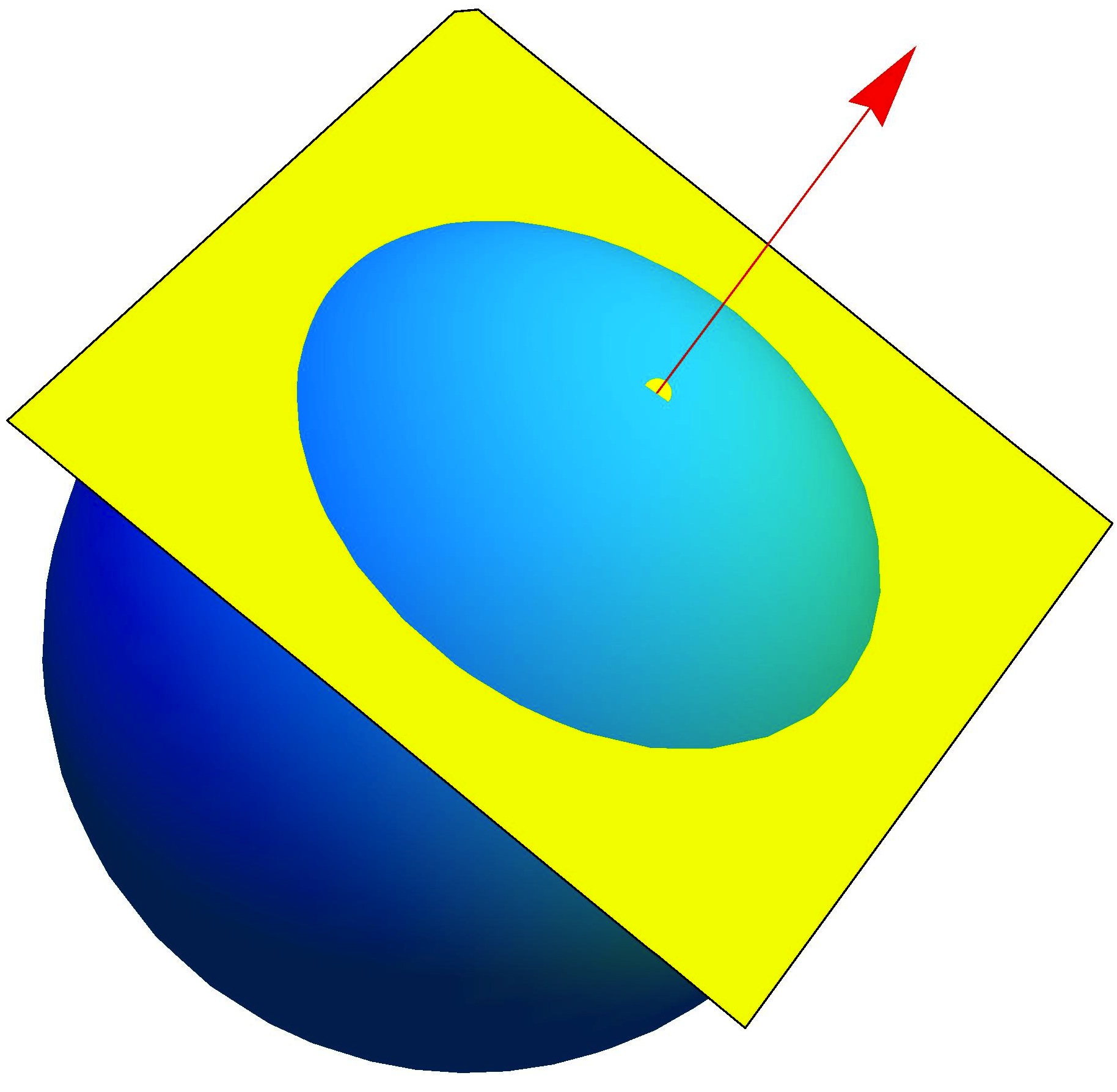}
\caption{A cap and its corresponding normal vector}
 \label{fig: cap with normal vector}
\end{subfigure}
\begin{subfigure}{0.5 \textwidth}
\includegraphics[width =  \textwidth]{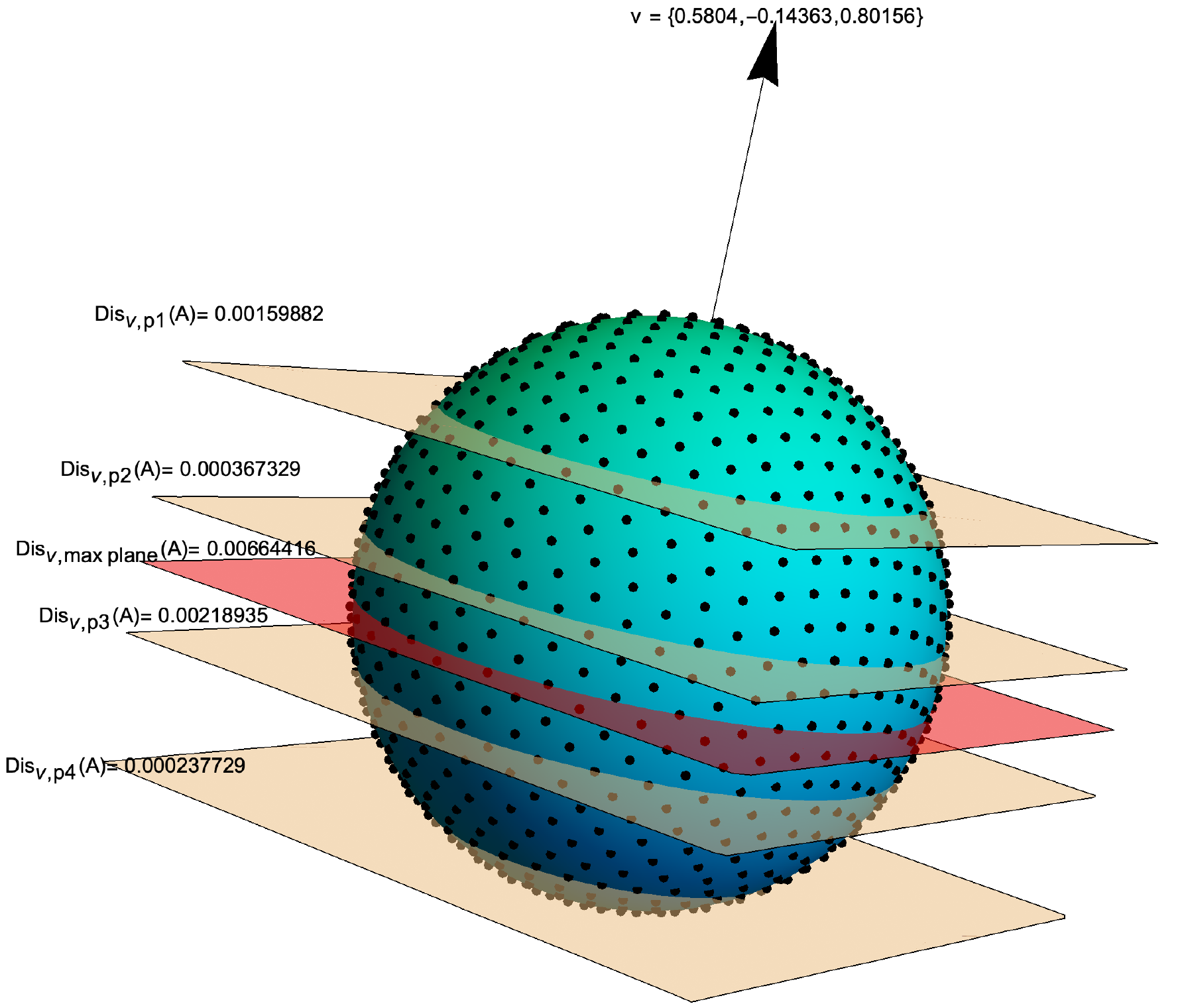}
\caption{Discrepancy of different caps with normal vector $\normalvec{v}$}
 \label{fig: caps discrepancy with same normal vector}
\end{subfigure}
\caption{directed Discrepancy}
\end{figure}
The following Lemma states that the upward jumps of the directed Discrepancy are bounded.  This Lemma is the main ingredient in the development of \alg algorithm. 

\begin{lem} \label{lem:confidence interval}
Let $P$ be a set of points with size $t$ on $\mathds{S}^2$, $\normalvec{v}$ a normal vector, and $ d \in \mathds{R}^+$ such that $Dis_{\normalvec{v}}(P) + \frac{1}{t} \leq d$. 
% I put this assumption because I need to choose two points x,y.  k should be greater than 2.
 Define 
 $P. \normalvec{v} = \cbr{\langle \mvec{p},\normalvec{v} \rangle \sVert[2] \mvec{p} \in P}  $,
 and suppose 
 $s_1 \leq s_2 \leq s_3 \leq ... \leq s_t$ 
 are sorted elements of $P.\normalvec{v}$.  Define
\begin{equation}
 k \define \lfloor t \intoo{ d - Dis_{\normalvec{v}}(P)} \rfloor + 1 ,
\end{equation} 
and 
\begin{equation} \label{eq:confidence radius}
r = \min \cbr{s_{i+k} - s_i \sVert[2] 1 \leq i \leq t - k}.
\end{equation}
Then, for all normal vectors $\normalvec{u}$ that satisfy the inequality $\envert{\normalvec{v} - \normalvec{u}} < r$, we have
 \begin{equation}
Dis_{\normalvec{u}}(P) \leq d.
 \end{equation}
\end{lem}

\begin{proof}
First of all, note that for each cap $C : \;  \langle \mvec{x},\normalvec{v} \rangle \geq h$, there is a cap with normal vector $\normalvec{u}$ who has the same area as $C$, and whose equation is $C' : \;  \langle \mvec{x},\normalvec{u} \rangle \geq h $.  For these two caps since $\mu (C) = \mu (C')$, the only difference in the right-hand side of equation \eqref{eq:Discrepancy definition wrt one set} is  $ \frac{1}{t} \envert{ \envert{C \cap P} - \envert{C' \cap P} } $, so we only need to show

\begin{equation} \label{eq:lem 1.1}
\frac{1}{t} \envert{ \envert{C \cap P} - \envert{C' \cap P} } \leq d - Dis_{\normalvec{v}}(P). 
\end{equation}

If we do so, by \eqref{eq:lem 1.1}, we obtain
\begin{align*}
{\rm Dis}_{C'}(P)
 & =
 \envert{ \frac{\envert{C' \cap P}}{t} - \frac{\mu (C')}{\mu (\mathds{S}^2)}  } 
 \\ & \leq
\envert{ \frac{\envert{C \cap P}}{t} - \frac{\mu (C)}{\mu (\mathds{S}^2)}  } +  \frac{1}{t} \envert{ \envert{C \cap P} - \envert{C' \cap P} }
\\ & 
\leq  Dis_{\normalvec{v}}(P) + ( d - Dis_{\normalvec{v}}(P)) 
\\ &
\leq d.
 \end{align*}

Because this holds for every cap $C'$ with normal vector $\normalvec{u}$ we get ${\rm Dis}_{\normalvec{u}}(P) \leq d$.

To prove \eqref{eq:lem 1.1}, note that points who belong to $C$ but not to $C'$ satisfy:
\begin{equation*}
    \langle \mvec{x}, \normalvec{v} \rangle \geq h, \quad
    \langle \mvec{x}, \normalvec{u} \rangle < h.
\end{equation*}
Note that since $t ( d - {\rm Dis}_{\normalvec{v}}(P)) \geq 1$, if  $ \envert{ \envert{C \cap P} - \envert {C' \cap P} } < 2$, \eqref{eq:lem 1.1} is automatically satisfied.  Now, let $\mvec{x}, \mvec{y} \in C - C'$ and without loss of generality assume $\mvec{y} . \normalvec{v} \leq \mvec{x} . \normalvec{v}$ so,
\begin{equation}
\langle \mvec{x}, \normalvec{u} \rangle < h \leq \langle \mvec{y}, \normalvec{v} \rangle \leq \langle \mvec{x}, \normalvec{v} \rangle,
\end{equation}
which implies
\begin{align} \label{eq:lem 1.2}
\envert{\langle \mvec{x}, \normalvec{v} \rangle - \langle \mvec{y}, \normalvec{v} \rangle} \leq \envert{\langle \mvec{x}, (\normalvec{v} - \normalvec{u}) \rangle} \leq   \norm{\mvec{x}} \norm{\normalvec{v} - \normalvec{u}} < r. 
\end{align}
Definition of $r$ implies the number of points of $P$ that may satisfy \eqref{eq:lem 1.2} is less than $k$, hence
\begin{equation}
\envert{P \cap (C - C') } \leq t ( d - Dis_{\normalvec{v}}(P)).
\end{equation}
Similarly, we can show $\envert{P \cap ( C' - C)} \leq t ( d - Dis_u(P)) $ so, the total difference between points on $C - C'$, which is difference of these two numbers, is less than $t ( d - Dis_u(P))$.
\end{proof}

\subsection{The Algorithm}
Lemma \ref{lem:confidence interval} is the main ingredient of the \alg algorithm.  The goal of the \alg is to check if $d$ is an upper bound for the Discrepancy of a given distribution of points.  The idea is to cover the surface of the unit sphere, as a representative of all normal vectors, with balls in which we know $d$ is an upper bound for the directed Discrepancy.  Since symmetry of some distributions allows us to consider only a portion of the Sphere, the algorithm is developed to cover any portion of the Sphere with bounds on the polar coordinates.  In other words, if we denote the polar coordinates of the normal vector $\normalvec{v}$ by $ \intoo{\theta, \phi} $, where
\begin{equation} \label{eq:polar coordinate definition}
\normalvec{v} = \intoo{\cos \phi \cos \theta, \cos \phi \sin \theta, \sin \phi}, \quad - \frac{\pi}{2} \leq \phi \leq \frac{\pi}{2}, \quad 0 \leq \theta < 2 \pi,
\end{equation}
then for any given $\phi_{\min}, \phi_{\max}, \theta_{\min}, \text{ and } \theta_{\max}$, the \alg checks if $d$ is an upper bound for directional Discrepancies of directions in the region $\phi_{\min} \leq \phi \leq \phi_{\max}, \quad \theta_{\min} \leq \theta \leq \theta_{\max}$.

\begin{rem} \label{rem:half sphere is enough}
Generally, studying directed Discrepancy for normal vectors on the upper half Sphere is enough since
\begin{equation} \label{eq:Discrepancy for -v}
{\rm Dis}_{- \normalvec{v}} (A) = {\rm Dis}_{\normalvec{v}} (A), \quad \forall \; \normalvec{v} \in \mathds{S}^2.
\end{equation}
In other words, we only need to consider the region $0 \leq \phi \leq \frac{\pi}{2}$.

Equation \eqref{eq:Discrepancy for -v} follows from the fact that the directed Discrepancy is the same for two complementary caps.
\end{rem}

Given $d$, we call $r$, defined in \eqref{eq:confidence radius}, the confidence radius around $\normalvec{v}$.  Algorithm
\ref{alg:directional discerpancy}  describes how the \alg  works.

\begin{algorithm}[h!] 
\SetAlgoLined
\KwResult{Covering family of the normal vectors, or a direction with Discrepancy greater or equal than $d - \frac{1}{t}$, or a set of not covered balls}
Input: $\phi_{\min}, \phi_{\max}, \theta_{\min},  \theta_{\max}, A$ \;
Initialization: $\phi = \phi_{\max}, \quad {\rm cannot\_cover} = \cbr{}$  \;
 \While{$\phi \geq \phi_{\min} $}{
	$r_{\min}$ = Approximate minimum confidence radius at latitude $\phi$ \;
 	$\theta = \theta_{\min}$ \;
	\While{ $ \theta \leq \theta_{\max} $ }{
	$\normalvec{v} = {\rm Polar}(\theta, \phi)$ \;
	$r_v =$ confidence radius at $\normalvec{v}$\;
	\If{$r_v < r_{\min}$}{
	Add $v$ to cannot-cover \;
	$r_v = r_{\min}$ \;
	}
	$\theta = \theta + d_{\theta}$ \;
	}
Check the lowest $\phi$ that is covered \;
Choose next $\phi$ accordingly \;
}
not-covered = $\emptyset$ \;
\While{\rm cannot\_cover $\neq \emptyset$}{
$ \normalvec{v} = {\rm pop (cannot\_cover )}$ \;
\eIf{Can cover ball with center $\normalvec{v}$ and radius $r_v$ with smaller balls }
{\textbf{continue}}
{Add $\normalvec{v}$ to not-covered}
}
\eIf{ \textit{not-covered} = $\emptyset$}{\textbf{return} Cover is completed}{\textbf{return} The algorithm cannot cover balls around vectors in \textit{not-covered}}
 \caption{ \alg \label{alg:directional discerpancy}}
\end{algorithm}

\subsection{Choice of the parameters for the \alg algorithm}
There are three parameters in the \alg that need to be tuned, $d_\theta$ at line 13, $\phi$ at line 16, and $r_{\min}$ at line 4.  We will describe how these parameters have been chosen in this section.

\begin{enumerate}
    \item $d_\theta$
    
    Having $\normalvec{v} = (\theta, \phi)$ and its confidence radius $r_{\normalvec{v}}$, we set $d_{\theta}$ such that the next direction, i.e. $(\theta + d_{\theta}, \phi)$ lies on the boundary of the confidence region of $\normalvec{v}$.  One, with a straightforward algebraic effort, can find the following closed form for this value.
    \begin{equation*}
        d_{\theta} = 2 \arcsin \intoo{\frac{r}{2 \cos \phi}}.
    \end{equation*}
    
    \item $\phi$
    
    With a binary search, we choose the next $\phi$ to be the lowest possible value such that two consecutive balls with radius $r_{\min}$ at the latitude intersects above the last cover latitude.  This is shown in Figure \ref{fig:next phi}.
    
    \begin{figure}[h]
        \centering
        \includegraphics[scale=.75]{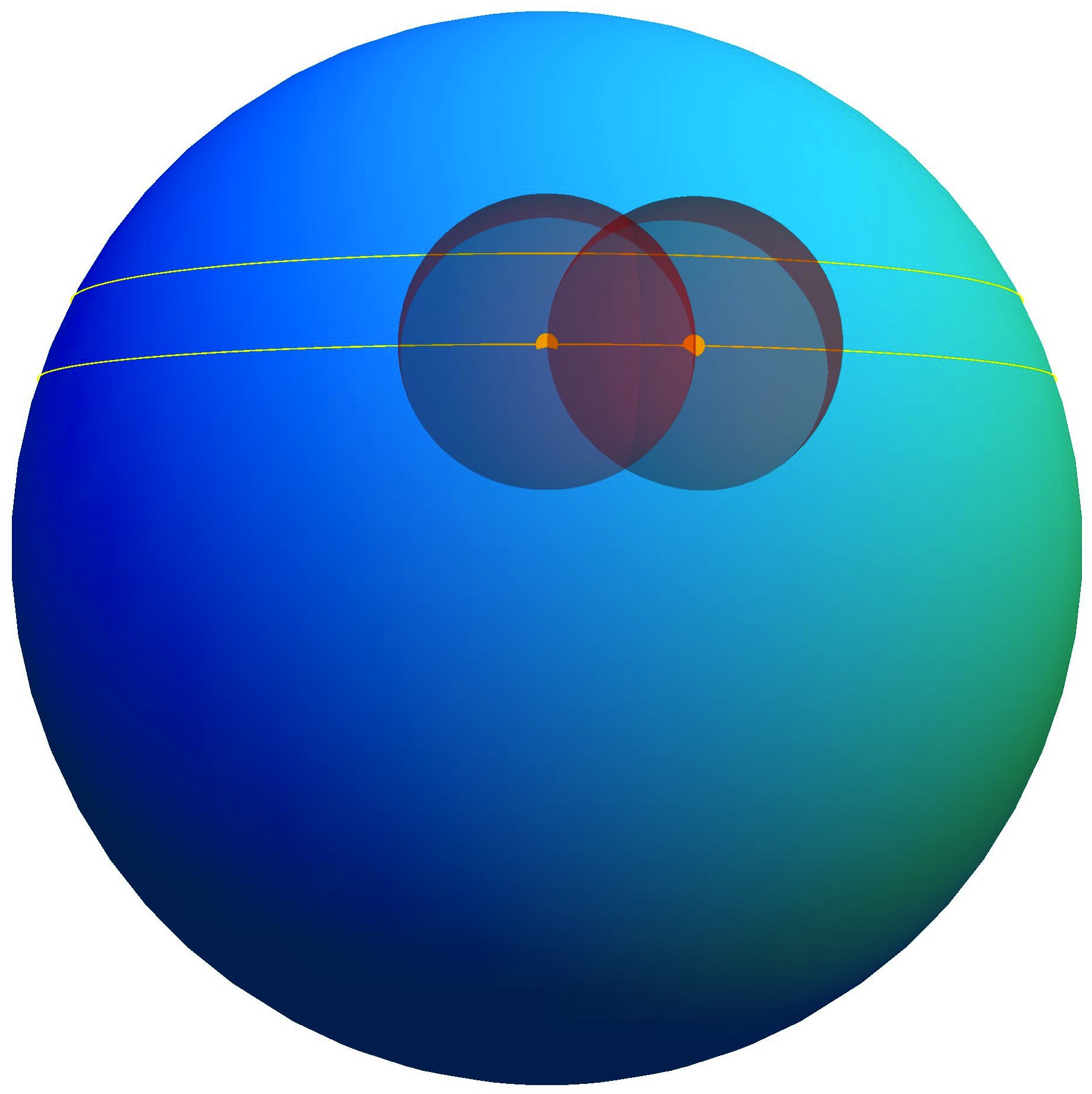}
        \caption{Choice of the new latitude based on the last covered one}
        \label{fig:next phi}
    \end{figure}
    
    \item $r_{\min}$
    
    At each latitude $\phi$, we take a sample of $20$ directions distributed uniformly at this latitude and calculate their confidence radius.  If ${\rm Med}(r)$ denotes the median of this data, we fix
    \begin{equation}
        r_{\min} = 0.5 {\rm Med}(r).
    \end{equation}
    We experimentally found that this setting makes the algorithm run faster.
\end{enumerate}

\subsection{A practical feature}

This section describes what happens in lines $19 - 25$ of Algorithm \ref{alg:directional discerpancy}. In order to prevent the Directional Discrepancy algorithm to get stuck in a direction with a notably smaller confidence radius in comparison to its neighbors, we approximate the average confidence radius  on an orbit through some samples and set a minimum acceptable radius for the given orbit, which excludes the outliers.  Through the run, if we get a direction whose confidence radius is less than the expected minimum we pass the direction to a module that covers the minimum expected radius around that direction with confidence radii around some close directions.  This module is called {Cover Cap} which is described in the following.

\subsubsection{Cover Cap}

As described above, the Cover Cap module aims to cover a cap around some outlier direction with the confidence caps around adjacent directions.  It is based on the following simple mathematical lemma.

\begin{lem} \label{lem: cover cap}
For all $ r \leq \sqrt{2} $, a cap of radius $r$ can be covered by $8$ caps of radius $\frac{r}{2}$, where one of them is centered at the center of the cap and the other $7$ caps are centered uniformly at the circle of distance $.86 r$ from the cap's center.  
\end{lem}

\begin{figure}[h]
    \centering
    \includegraphics[scale=.75]{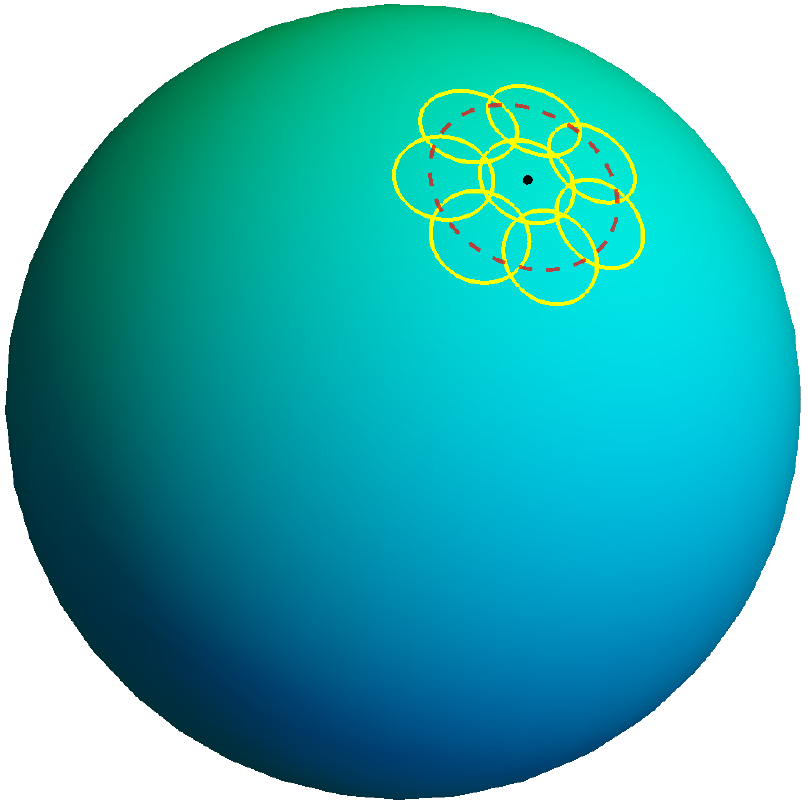}
    \caption{Coverage of a cap of radius $r$ with 8 caps of radius $\frac{r}{2}$}
    \label{fig: cover cap}
\end{figure}

\begin{proof}
Without loss of generality, we may assume the cap is centered at the North Pole, $\intoo{0, 0, 1}$.  It suffices to show the upper intersection of two consecutive caps at distance $.86 r$ is inside the cap around the North Pole with radius $\frac{r}{2}$ and the lower intersection lies out of the cap with radius $r$ around the North Pole.

One can easily confirm that the circle of distance $.86 r$ is located at the polar coordinates equation $\varphi = \arcsin \intoo{\frac{2 - (.86 r)^2}{2}}$.  Hence we may choose the consecutive cap centers $\normalvec{v}_1 = (0, \varphi)$ and $ \normalvec{v}_2 = \intoo{\frac{2 \pi}{7}, \varphi}$ in polar coordinates.  Let call the upper and lower intersections of caps centered at $\normalvec{v}_1, \normalvec{v}_2$; $\mvec{u}$ and $\mvec{d}$, respectively.

One can easily check both $\mvec{u}$ and $\mvec{d}$ satisfy polar equation
\begin{equation}
    \theta = \frac{\pi}{7}.
\end{equation}
Moreover, we have $\norm{\normalvec{v}_1 - \mvec{u}} =  \frac{r}{2} $ or equivalently, $ \langle \normalvec{v}_1, \mvec{u} \rangle = 1 - \frac{r^2}{8} $, which yields 

\begin{align} \label{eq: sin phi equation}
    \cos \varphi \cos \frac{\pi}{7} \cos \phi_u + \sin \varphi \sin \phi_u = 1 - \frac{r^2}{8},
\end{align}
where $\intoo{\frac{\pi}{7}, \phi_u}$ denotes the polar coordinates of $\mvec{u}$, and $\sin \varphi = 1 - \frac{.86 ^ 2}{2} r^2$.  Similarly, $\mvec{d}$ satisfies equation \eqref{eq: sin phi equation}.  Hence, if we replace $\cos \phi_u = \sqrt{1 - \sin^2 \phi_u} $ and solve the quadratic equation \eqref{eq: sin phi equation} for $x \define \sin \phi_u$, we obtain
\begin{equation} \label{eq: sin phi closed form}
    \sin \phi_{\mvec{u}}, \sin \phi_{\mvec{d}} =
    \frac{1.796 r^4 + \left( \pm 1.881 \sqrt{-0.4790 r^4-2.653 r^2+28.36}-19.22\right) r^2+38.84}{r^4-5.408 r^2+38.84}.
\end{equation}

To complete the proof, we need to show
\begin{equation*}
    \sin \phi_{\mvec{u}} \geq 1 - \frac{r^2}{8}, \quad \sin \phi_{\mvec{d}} \leq 1 - \frac{r^2}{2}, 
\end{equation*}
for $r \in \intoc{0, \sqrt{2}}$.  While one can algebraically prove these inequalities hold, we think the following figure is more informative to do so.

\begin{figure}[h]
    \centering
    \includegraphics[scale = .7]{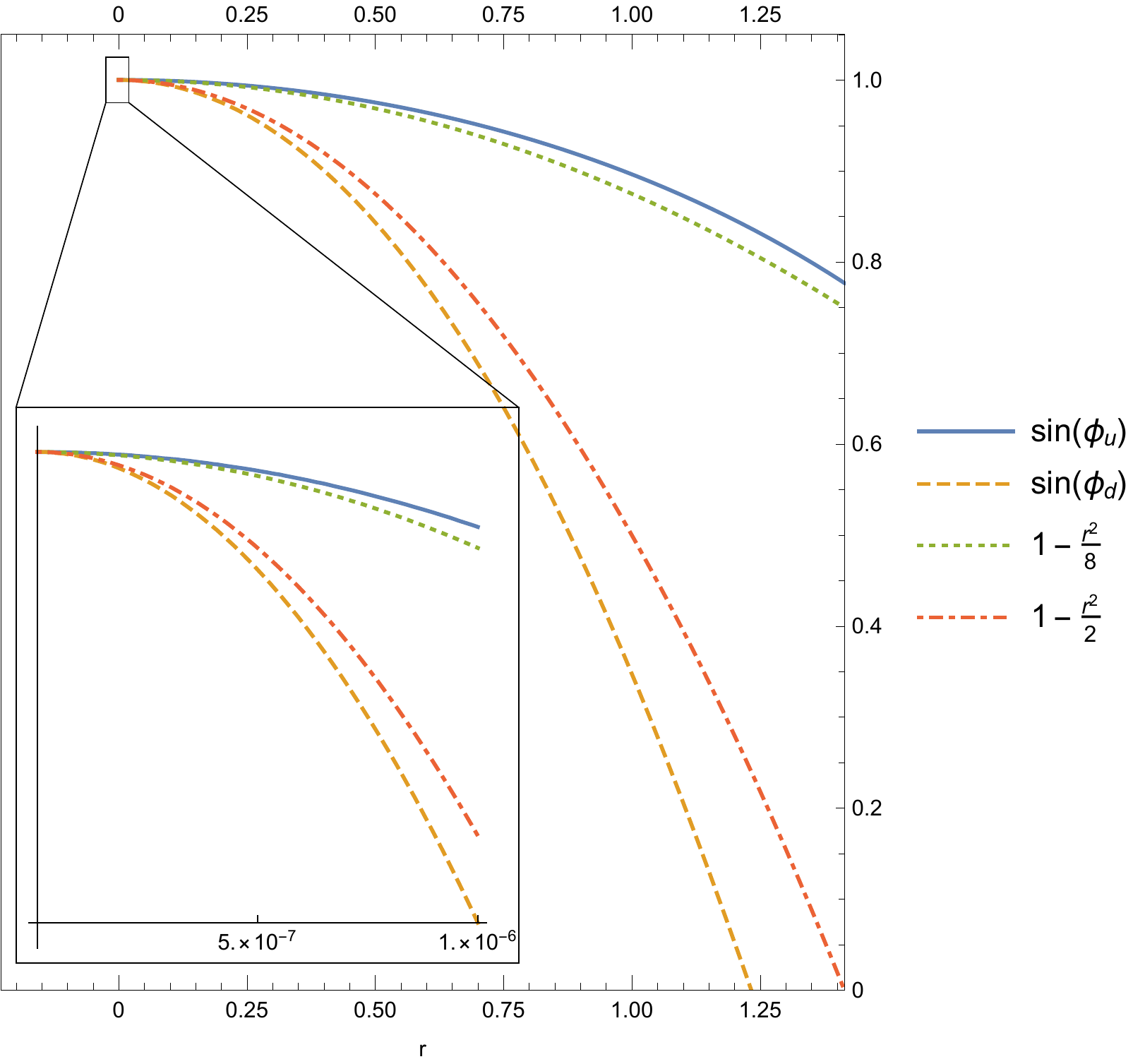}
    \caption{Relation of the roots of \eqref{eq: sin phi equation}, $ 1 - \frac{r^2}{2}$, and $1 - \frac{r^2}{8}$ }
    \label{fig:sin equation roots and bounds}
\end{figure}

As Figure \ref{fig:sin equation roots and bounds} denotes, plots of $sin(\phi_{\mvec{u}}), sin(\phi_{\mvec{d}})$ and the required bounds touch each other at $(0, 1)$, and for $r \in \intoc{0, \sqrt{2}}$ we have
\begin{equation*}
    \sin(\phi_{\mvec{u}}) > 1 - \frac{r^2}{8} > 1 - \frac{r^2}{2} > \sin(\phi_{\mvec{d}}),
\end{equation*}
which concludes the proof.
Figure \ref{fig: cover cap} denotes the coverage of a cap with smaller caps.
\end{proof}

If the confidence radius at a vector $\normalvec{v}$ is less than the minimum expected radius $r_{\rm min}$, We check if the confidence radius at each of the centers given in Lemma \ref{lem: cover cap} is larger than $\frac{r_{\rm min}}{2}$.  If so, we increase the confidence radius at $\normalvec{v}$ to $r_{\rm min}$.  Moreover, if the confidence radius at any center is less than $\frac{r_{\rm min}}{2}$ we can repeat this machinery.  This recursive procedure will stop after a preset number of radius checking.  If it fails to cover the expected radius, it will return the small uncovered area as a region where one might find a larger Discrepancy or some singularity in confidence radius.  This sub-module is called {Cover Cap}.  Note that Cover Cap also serves as a searching method for finding directions with high Discrepancy and will report the largest Discrepancy it found at all the singular caps.

\subsection{Time complexity}
It is not hard to show for any distribution of points like $A \subset \mathds{S}^2$ there are finitely many caps with the possibility of having the maximum Discrepancy.  Planes that cut those caps should pass through a few points from $A$.  However, the number of these candidate caps is $O(\abs{A}^3)$.  Moreover, calculating the Discrepancy of each cap has the time complexity $O(\abs{A})$, hence in total such method will take $O(\abs{A}^4)$ to run.
Nevertheless, \alg can provide us with upper bounds with arbitrary but fixed deviation from the actual Discrepancy, say $0.001$ for instance, in $ O(\abs{A}) $. 

Note that, since Lemma \ref{lem:confidence interval} is the main tool used in the \alg algorithm, it cannot confirm upper bounds that are close to the actual Discrepancy more than $ 1 / \abs{A} $.  However, this is not restrictive to get arbitrarily close to the true answer as the size of $A$ grows to infinity.  The following Theorem gives an upper bound for the time complexity of the \alg algorithm.

\begin{thm} \label{thm:time complexity}
Assume we are interested to check if the directed Discrepancy in a region $D \subset \mathds{S}^2$ is less than $d$.  Let $d_{\max}$ denotes the maximum directed Discrepancy in this region. Let $A$ be the set of points distributed on the $\mathds{S}^2$, and denote its size by $t$.  If the following assumptions are satisfied:
\begin{enumerate}
\item \label{asumption: d > D_max + c}
$ d > d_{\max} + \frac{h - 1}{t},$ for some positive integer $h$,
\item \label{asumption:projections not too closed}
Let $\langle A, \normalvec{v} \rangle_{(1)} \leq \langle A, \normalvec{v} \rangle_{(2)} \leq ... \leq \langle A, \normalvec{v} \rangle_{(T)}$ be the sorted set of projection of $A$ on the normal vector $\normalvec{v}$. Then assume
\begin{equation} \label{eq:def of r_min}
r_{\min} \define \min\limits_{\substack{1 \leq i \leq t - h \\ \normalvec{v} \in D }} \langle A, \normalvec{v} \rangle_{(i + h)} - \langle A, \normalvec{v} \rangle_{(i)} > 0,
\end{equation}

\end{enumerate}
Then, the \alg can cover the whole region $D$, and its run time is $ O \intoo{ \frac{t}{r_{\min}^2} } $.
\end{thm}

\begin{proof}
Without losing generality, we assume the North pole is included in the region $D$.  The proof for the other case is analogous. 
Note that at latitude $\phi$ we have $d_{\theta} \geq 2 \arcsin \intoo{\frac{r_{\min}}{2 \cos \phi}} \geq r_{\min}$, hence at each orbit, we put at most $\frac{2 \pi}{r_{\min}}$
directions to cover it completely.
Moreover, let $\normalvec{v}_1 = (0, \phi)$ and $v_2 = (d_\theta, \phi)$ be two consecutive centers at this orbit whose distance is $r_{\min}$, which means
\begin{equation}
d_\theta = 2 \arcsin \intoo{\frac{r_{\min}}{2 \cos \phi}}.
\end{equation}
We want to upper bound the latitude of the lower intersection of balls with radius $r_{\min}$ centered at $\normalvec{v}_1, \normalvec{v}_2$ to bound the total number of orbits.  Denote the polar coordinates of such intersection with $\mvec{u} = (\tau, \psi)$.  Therefore we have
\begin{equation} \label{eq:relation of v1 v2 and u}
\norm{\normalvec{v}_1 - \mvec{u}} = \norm{\normalvec{v}_2 - \mvec{u}} = \norm{\normalvec{v}_1 - \normalvec{v}_2} = r_{\min}, \quad \psi < \phi.
\end{equation}

It is straightforward to deduce from \eqref{eq:relation of v1 v2 and u} that
\begin{equation} \label{eq:inner product of vi and u}
\langle \normalvec{v}_1, \mvec{u} \rangle = \langle \normalvec{v}_2, \mvec{u} \rangle = 1 - \frac{r_{\min}^2}{2},
\end{equation}
thus
\begin{equation}
\langle \normalvec{v}_1 - \normalvec{v}_2, \mvec{u} \rangle = 0.
\end{equation}
Hence,
\begin{align*}
0 &= 
(\cos d_\theta - 1) \cos \tau + \sin d_\theta \sin \tau
= \cos( d_\theta - \tau ) - \cos \tau,
\end{align*}
which yields
\begin{equation} \label{eq:tau is dtheat/2}
\tau = \frac{d_\theta}{2} = \arcsin \intoo{\frac{r_{\min}}{2 \cos \phi}}.
\end{equation}

Note that since $r_{\min}$ is a lower bound for confidence radius of all directions in $D$, including the North Pole, the maximum $\phi$ that we need to cover is equal or less than $\arcsin \intoo{1 - \frac{r_{\min}^2}{2}}$.

Now, by \eqref{eq:relation of v1 v2 and u}, we obtain
\begin{align} \nonumber
1 - \frac{r_{\min}^2}{2} &=
\langle \normalvec{v}_1, \mvec{u} \rangle
\\ & = \nonumber
\cos \phi \cos \psi \cos \tau + \sin \phi \sin \tau
\\ & = \nonumber
\cos (\phi - \psi) + (\cos \tau - 1) \cos \phi \cos \psi
\\ & \geq \nonumber
\cos (\phi - \psi) + (\cos \tau - 1) \cos^2 \phi 
\\ & = \label{eq:thm time 1}
\cos (\phi - \psi) + \intoo{\sqrt{1 - \frac{r_{\min}^2}{4 \cos^2 \phi}} - 1} \cos^2 \phi. 
\end{align}
The last equality is derived by using \eqref{eq:tau is dtheat/2}.  Note that
\begin{equation*}
\intoo{1 - \sqrt{1 - \frac{r_{\min}^2}{4 \cos^2 \phi}}} \cos^2 \phi \leq \frac{r_{\min}^2}{4},
\end{equation*}
therefore by \eqref{eq:thm time 1} we obtain
\begin{equation*}
\cos (\phi - \psi) \leq 1 - \frac{r_{\min}^2}{4},
\end{equation*}
or in other words,
\begin{equation} \label{eq:psi upper bound}
\phi - \psi \geq - \arccos \intoo{1 - \frac{r_{\min}^2}{4}}.
\end{equation}
This means we need at most 
\begin{equation} \label{eq:upper bound for number of orbits}
\frac{\frac{\pi}{2}}{\arccos \intoo{1 - \frac{r_{\min}^2}{4}}}
\leq
 \frac{\pi}{ \sqrt{2} r_{\min}}
\end{equation}
orbits to cover the semi-Sphere.  The last inequality is deduced from the inequality $ \arccos (1 - x) \geq \sqrt{2 x} $, which is immediate result by the concavity of $\arccos$ in $[0, 1]$.

Since that obtaining the directed Discrepancy and the confidence radius at any given normal vector is an $O(t)$ computation, given the number of directions at each orbit is at most $\frac{2 \pi}{r_{\min}}$, and \eqref{eq:upper bound for number of orbits}, we deduce the time complexity of \alg algorithm is at most $ O \intoo{\frac{t}{r_{\min}^2}} $, as desired.

\end{proof}

\begin{rem} \label{rem: cover cap time complexity}
Since the number of recursions of the Cover Cap at each direction is upper bounded by a constant, the time complexity of the Cover Cap in total is $n_{CC} O(t)$, where $n_{CC}$ denotes the number of directions for which the Cover Cap was called.  Hence, the Cover Cap module does not make any change in final time complexity which is $n_{DD} O(t)$, where $n_{DD}$ is the total number of directions needed to cover the area of interest.  If one sets the expected radius at each orbit such that $\frac{n_{CC}}{n_{DD}} = o(1)$, then the run time of the Cover Cap will be negligible.
\end{rem}

\begin{rem} \label{rem:r <= 2/t}
Note that since there are $t$ projected points, we always have 
\begin{equation} \label{eq:r upper bound}
r \leq \frac{2 h}{t},
\end{equation}
regardless of the distribution.
\end{rem}

Note that $r_{\min}$ is closely related to the Discrepancy in the sense that small $r_{\min}$ shows a dense concentration of points of the set $A$ around some affine plane, hence when Discrepancy of a distribution is small we expect to have large $r_{\min}$. If for a given distribution we have $r_{\min} \geq \frac{c h}{t}$, i.e. $r_{\min} = \Omega(d - d_{\max})$, then Theorem \ref{thm:time complexity} states the time complexity is $O \intoo{\frac{t}{(d - d_{\max})^2}}$.
In this case, the following remarks offer an upper bound on the run time of the Directional Discrepancy algorithm in terms of $t$.  For smaller $ r_{\min} $, the run time gets larger accordingly which can be obtained in a similar manner.

\begin{rem} \label{rem:t = O(t)}
Suppose $r_{\min} = \Omega(d - d_{\max})$, and for a positive $\epsilon$ we set $\epsilon = d - d_{\max} $, then by Theorem \ref{thm:time complexity}, the \alg runs in $O(t)$ and can prove upper bound $ d_{\max} + \epsilon$ when the Discrepancy is $d_{\max}$.
\end{rem}

\begin{rem}
%\milad{Mina, please do the calculations of this remark and make sure there is no mistake here.  Also, we need experimental evidence for the assumption $r_{\min} = O(h/t)$.  Please include the evidence in the section you are writting.}
Assume we are in the same setting as Remark \ref{rem:t = O(t)}. When $d_{\max} = O(t^{\frac{-1}{2}})$, which is the case, for best known deterministic distributions, by letting $ d - d_{\max} = \epsilon t^{\frac{-1}{2}} $ the \alg can prove the upper bound of order $d_{\max} + \epsilon t^{-\frac{1}{2}}$.  In this case, the running time is at most $O(t^2)$.

Similarly, when $d_{\max} = O \intoo{t^{-\frac{3}{4}}}$, the \alg algorithm can prove upper bounds with a small relative error while its time complexity is $O(t^{\frac{5}{2}})$.  Recall that from \cite{beck-1984}, we know the minimum possible Discrepancy is $O(t^{-\frac{3}{4}})$.  This means the \alg is capable of confirming upper bound that have a small-scaled relative deviation from the exact Discrepancy at a cost of extra $t^{\frac{3}{2}}$ factor in the run time.
\end{rem}

\section{An application} \label{sec:an application}

\subsection{Polar Coordinates} \label{subsec:polar coordinates}
Polar Coordinates has been introduced in \cite{katanforoush2003} as a distribution of points to optimize the Coulomb potential.  It distributes points on the Sphere as follows.  Given a positive integer $n$, Polar Coordinates puts a couple of points at two poles and in $n - 1$ orbits, whose latitudes are given by
\begin{equation} \label{eq:phis of polar coordinates}
\phi_j = \frac{\pi j}{n} - \frac{\pi}{2}, \quad j = 1, ..., n-1.
\end{equation}
Moreover, define $n_j = \lfloor \frac{1}{2} + \sqrt{3} \cos \phi_j \rfloor$ for $1 \leq j \leq n - 1$.  At $j^{\rm th}$ latitude, it has $n_j$ equally spaced points starting at $\theta = 0$.  Then they shift alternate latitude to achieve a more symmetric distribution.

During our study on Discrepancy, we found the Polar Coordinates can be a candidate for one of the lowest known Discrepancy among deterministic distributions on the Sphere.  First of all, we have shown the following,

\begin{thm}\label{thm:polar coordinates North Pole dis}
If Polar Coordinates puts $t$ points on the Sphere, for the North Pole $= \intoo{0, 0, 1}$ the directed Discrepancy is $O(t^{-\frac{1}{2}})$.
\end{thm}

\begin{proof}
Let denote the Polar Coordinates with index $n$ by $P_n$.  As discussed in \cite{katanforoush2003}, we have 
\begin{equation*}
    t = \abs{P_n} \simeq \Theta(n^2).
\end{equation*}
It is straightforward to check that the maximum Discrepancy is attained by a cap that passes through at least one point of $P_n$.  Since orbits of Polar Coordinates are on planes with normal vector at the North Pole, we just need to bound the Discrepancy at each cap whose boundary is an orbit of $P_n$.  Moreover, because of the symmetry of $P_n$ to $x y$ plane, it suffices to consider orbits with $ z \geq 0$.  Let $n_j$ denotes the number of points at $j^{\rm th}$ orbit of $P_n$, then the number of points on the cap above the $j^{\rm th}$ orbit is
\begin{align} \label{eq: points above cap j}
    S_j \define & =
    \sum_{i = j}^n n_i = 1 + \sum_{i = j}^{n - 1} \lfloor \frac{1}{2} + \sqrt{3} n \cos \phi_i \rfloor = \sqrt{3} n \sum_{i = j}^{n - 1}    \cos \phi_i + f(j),
\end{align}
where $\abs{f(j)} \leq 2 n$.
Moreover, note that

\begin{equation*}
    \sum_{i = j}^{n - 1}\cos \phi_i =
    \sum_{i = j}^{n - 1}\cos \intoo{\frac{\pi i}{n} - \frac{\pi}{2}}
    = \frac{\cos \intoo{\frac{\pi}{2 n}} + \cos \intoo{\frac{(2j - 1) \pi}{2n}}}{2 \sin \intoo{\frac{\pi}{2 n}}}.
\end{equation*}

Hence, 
\begin{equation} \label{eq: points above cap j simplified}
    S_j = \sqrt{3} n \frac{\cos \intoo{\frac{\pi}{2 n}} + \cos \intoo{\frac{(2j - 1) \pi}{2n}}}{2 \sin \intoo{\frac{\pi}{2 n}}} + f(j), \quad \abs{f(j)} \leq 2n,
\end{equation}
and in particular

\begin{equation} \label{eq: t simplified}
    t = S_0 = \sqrt{3} n \frac{ \cos \intoo{ \frac{\pi}{2 n}}}{ \sin \intoo{\frac{\pi}{2 n}}} + f(0), \quad \abs{f(0)} \leq 2n.
\end{equation}

Let denote the cap above $j^{\rm th}$ orbit by $C_j$.  Note that since $j^{\rm th}$ orbit is located at $z = \sin \phi_j$, we have
\begin{equation}
    \frac{\mu (C_j)}{\mu (\mathds{S}^2)} = \frac{1 - \sin \phi_j}{2}
    = \frac{1 + \cos \frac{\pi j}{2}}{2}.
\end{equation}
Therefore,
\begin{align} \nonumber
    \abs{ \frac{S_j}{t} -  \frac{\mu (C_j)}{\mu (\mathds{S}^2)} } & =
    \frac{1}{t} \abs{ S_j - \frac{1 + \cos \frac{\pi j}{2}}{2} t }
    \\ & \nonumber
    = \frac{1}{t} \abs{ \sqrt{3} n \frac{\cos \intoo{\frac{\pi}{2 n}} + \cos \intoo{\frac{(2j - 1) \pi}{2n}}}{2 \sin \intoo{\frac{\pi}{2 n}}} + f(j)  - \frac{1 + \cos \frac{\pi j}{n}}{2} \intoo{\sqrt{3} n \frac{ \cos \intoo{ \frac{\pi}{2 n}}}{ \sin \intoo{\frac{\pi}{2 n}}} + f(0)}}
    \\ & \nonumber
    \leq \frac{1}{t} \frac{n \sqrt{3}}{2} \abs{ \frac{\cos \frac{\pi}{2 n} + \cos \frac{(2 j - 1) \pi}{2 n} - (1 + \cos \frac{\pi j}{n}) \cos \frac{ \pi}{2 n}}{\sin \frac{\pi}{2 n}} } + \frac{\abs{f(j)} + \abs{f(0)}}{t}
    \\ & \nonumber
    \leq \frac{n \sqrt{3}}{2 t} \abs{ \frac{\cos \intoo{\frac{\pi j}{n} - 
    \frac{\pi}{2n}} - \cos \frac{\pi j}{n} \cos \frac{\pi}{2 n}}{\sin \frac{\pi}{2 n}} } + \frac{4 n}{t}
    \\ & \nonumber
    = \frac{n \sqrt{3} \abs{\sin \frac{\pi j}{n}}}{2 t} + \frac{4 n }{t}
    \\ & \nonumber
    \leq \intoo{\frac{\sqrt{3}}{2} + 4} \frac{n}{t}
    \\ & 
     = O(t^{-\frac{1}{2}}).
\end{align}
The last equality holds since we have $t = \Omega(n^2)$.

\end{proof}

Moreover, we observed that for any $n$ that we could make the Polar Coordinates and run the Directional Discrepancy, North Pole has the highest directed Discrepancy.  In other words, the Discrepancy of Polar Coordinates is the same as the North Pole directed Discrepancy.  Note that Polar Coordinates has a special symmetry around the North Pole. 
This makes us propose the following conjecture and seek for evidence to prove it via the \alg algorithm.

\begin{conj} \label{conj:North Pole has the max Discrepancy}
For any $n \in \mathds{Z}^+$, Polar Coordinates with parameter $n$ has maximum directed Discrepancy in the direction of the North Pole.
\end{conj}

We can prove the North Pole is, at least, a local maxima for the directed Discrepancy. This is what the following Lemma is concerned with.

\begin{lem} \label{lem: north pole confidence radius}
For Polar Coordinates set $P_n$, there is a positive radius $r$ such that the North Pole has the maximum directed Discrepancy in its neighborhood of distance $r$.
\end{lem}

\begin{proof}
The Lemma follows from the fact that if a normal vector $\normalvec{v}$ is closer to the North Pole than a certain distance, a plane perpendicular to $\normalvec{v}$ cannot intersect with more than one orbit of $P_n$.  To find such a radius, one needs to take the minimum over the distances of the normal vectors of the planes that connect the edges of two consecutive orbits of $P_n$, as figure \ref{fig: North Pole confidence radiud} depicts.  More precisely, for orbits with equations $ z = z_1, z = z_2, \quad 0 \leq z_1 < z_2 $, one needs to consider the diagonal plane of the Cuboid with a face circumscribed the orbit at $z = z_1$ and a parallel face passing through $z  = z_2$.  If the normal vector of a plane $Pl$ has a lower distance to the North Pole than the normal vector of this diagonal plane, then $Pl$ cannot intersect both of the above orbits.

Since we have finitely many normal vectors for consecutive orbits, the minimum of their distance to the North Pole would be a positive number $r$.

If a cap is cut by a plane that intersects with at most one orbit, and hence the plane passes through at most two points of $P_n$, one can rotate the plane such that the number of points from $P_n$ remains constant on the cap while the cap area changes in a direction that increases the Discrepancy (One is able to both increase and decrease the cap area).  Hence, he can find a cap perpendicular to the North Pole with directed Discrepancy equal to or greater than the original cap.  This concludes the proof.

\begin{figure}[h] 
    \centering
    \includegraphics[scale = .75]{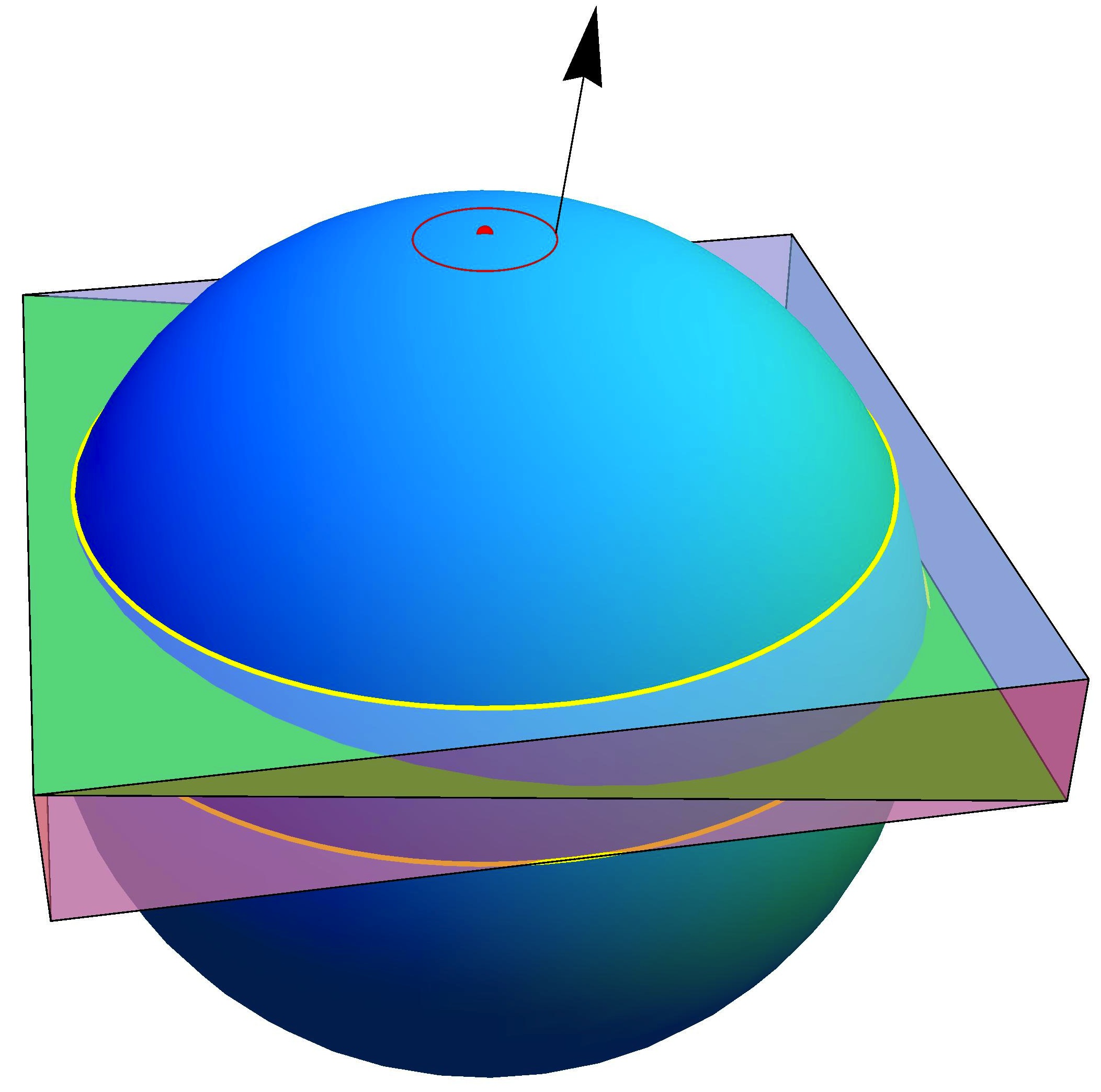}
    \caption{Normal vector of a plane connecting two consecutive orbits }
   \label{fig: North Pole confidence radiud}
\end{figure}
\end{proof}

While Conjecture \ref{conj:North Pole has the max Discrepancy} seems to be true for the original Polar Coordinates distribution, we found the \alg runs easier if we add a different shift at each orbit.  Note that, shifting an orbit does not change the Discrepancy in the direction of the North Pole, but it provides us with a larger confidence radius around $\normalvec{v} = \intoo{0, 1, 0}$. The reason is the original Polar Coordinates puts many points on a plane that are orthogonal to $\intoo{0, 1, 0}$ which can make the confidence radius in this direction so small.  We call this shifted version of the distribution {Twisted Polar Coordinates}.

\begin{defn} \label{def:twisted polar coordinates}
Twisted Polar Coordinates has exactly the same procedure to distribute points on the Sphere as Polar Coordinates defined in subsection \ref{subsec:polar coordinates}.  It has a different shift at each orbit with is given by
\begin{equation}
\text{shift}_j = \frac{j}{n} \frac{2 \pi }{n_j}, \quad 1 \leq j < \lfloor \frac{n_j}{2} \rfloor.
\end{equation}
We also shift the orbits below the plane $z = 0$ such that we keep symmetry with respect to this plane.  Figure \ref{fig:twisted pc} compares Polar Coordinates with Twisted Polar Coordinates for $n = 14$. 
\end{defn}

\begin{figure}[h]
    \centering
    \includegraphics[scale = .75]{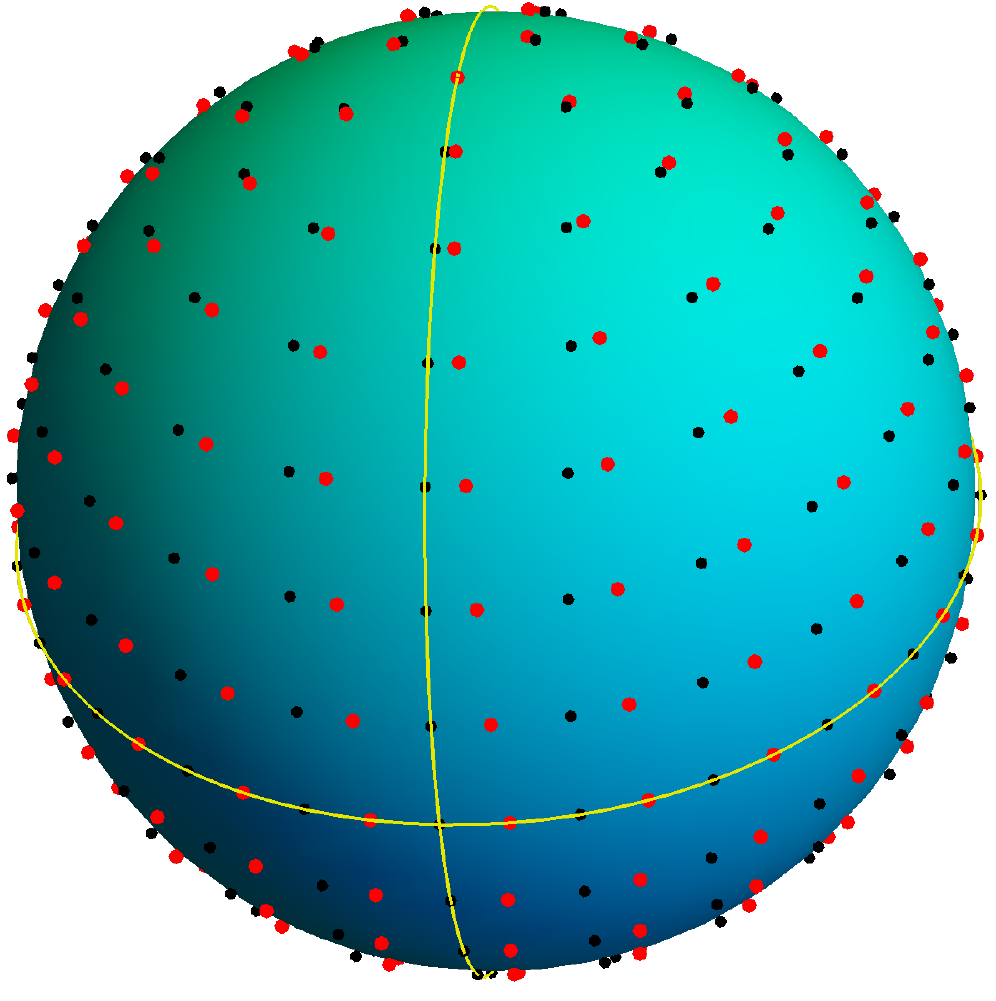}
    \caption{Polar Coordinates (small black points) vs. Twisted Polar Coordinates (large red points) for $n = 14$.}
    \label{fig:twisted pc}
\end{figure}

\begin{rem}
Note that Polar Coordinates and Twisted Polar Coordinates have the same directed Discrepancy at the North Pole.
\end{rem}

\subsection{Supporting evidence}

We have used both the naive search, when there are a few points on the Sphere, and the Directional Discrepancy algorithm to testify Conjecture \ref{conj:North Pole has the max Discrepancy}.

For $2 \leq n \leq 15$, we could run the naive method on $P_n$ and explicitly check the North Pole has the highest Discrepancy.  For $15 \leq n \leq 125$, we took advantage of the Directional Discrepancy algorithm in the following way.
By Lemma \ref{lem: north pole confidence radius} we could find $\phi_{\max}$ such that for any direction $\normalvec{v}$, with $\phi \geq \phi_{\max}$ in polar coordinates, the directed Discrepancy at $\normalvec{v}$ is less than the one in the direction of the North Pole.  Hence, we just need to check the directed Discrepancy for the rest of the directions.  Note that, by definition, the directed Discrepancy is the same in directions $\normalvec{v}$ and $ - \normalvec{v}$ for any given direction $\normalvec{v}$.  Thus, it suffices to consider directions on the upper semi-Sphere, i.e. for $z \geq 0$.  Furthermore, since the set $P_n$ is symmetric to the $xy$ plane, directed  Discrepancy at directions $(x, y, z), (- x, -y, z)$ and $(-x, -y, -z)$ are all the same, therefore it suffices to consider directions with $0 \leq \theta \leq \pi$ in the polar coordinates.

Given all the above, our strategy to check Conjecture \ref{conj:North Pole has the max Discrepancy} is as follows.  For $P_n$, first, we find $\phi_{\max}$ based on Lemma \ref{lem: north pole confidence radius}.  Then, we set $d$ to be the directed Discrepancy at the North Pole and run the Directional Discrepancy algorithm to check if the directed Discrepancy in the region $0 \leq \phi \leq \phi_{\max}, \quad 0 \leq \theta \leq \pi$ is less than $d$.  The affirmative answer to this question rigorously proves the Conjecture \ref{conj:North Pole has the max Discrepancy} for tested cases.  This was the case for all $15 \leq n \leq 125$
for which we run the Directional Discrepancy. Figure \ref{fig: coverage for n = 60} shows the directions and their confidence radii with which the Directional Discrepancy covers the area of interest for $P_{60}$.  In this picture, North Pole has been shown by a red vector and the boundary of the region at which we proved it is local maxima is depicted by a yellow circle.

\begin{figure}[h]
    \centering
    \includegraphics{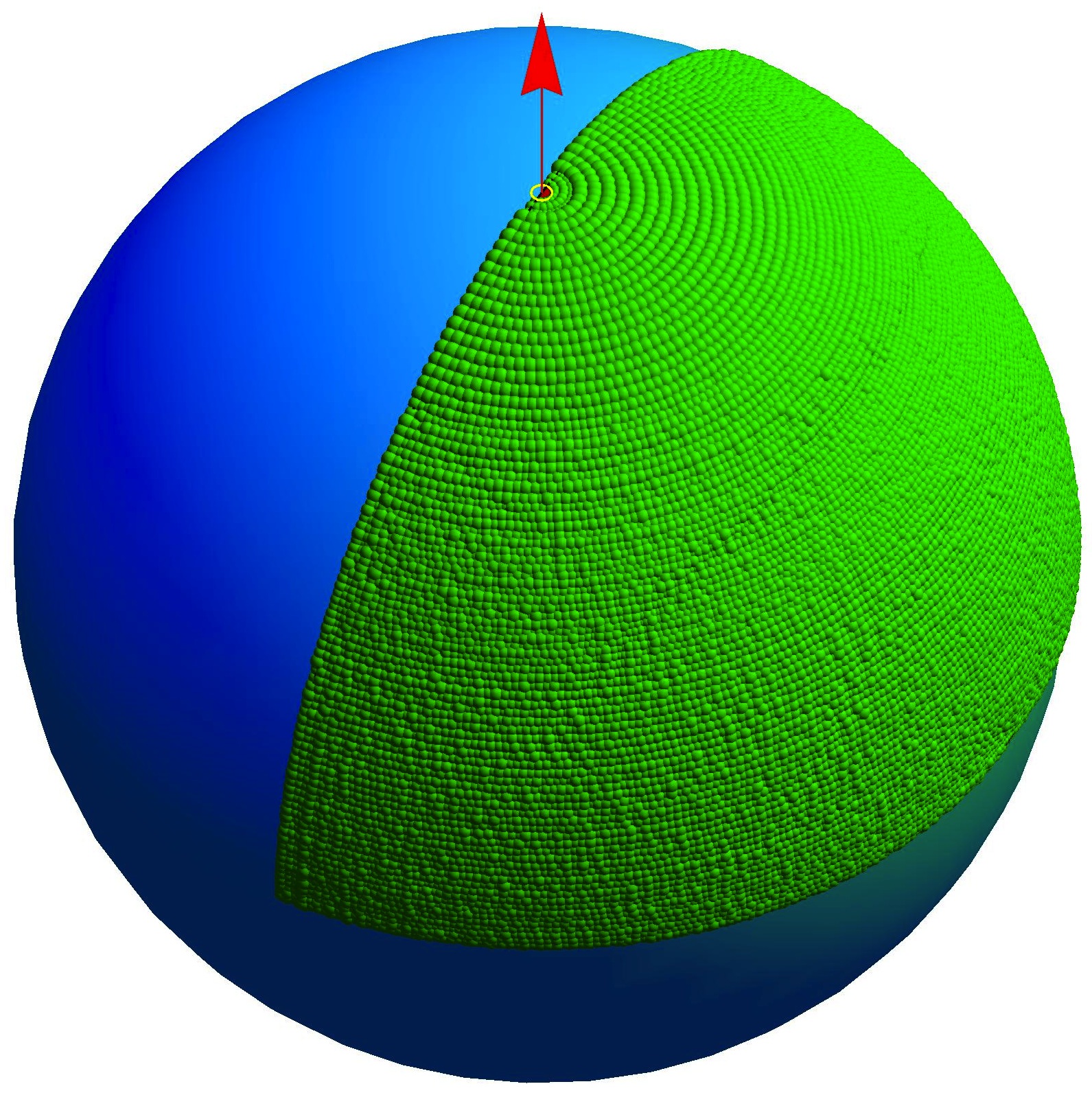}
    \caption{Coverage of the required directions with confidence radii for $n = 60$}
    \label{fig: coverage for n = 60}
\end{figure}

The data of this run is available at the following link: 

\href{https://github.com/minalatifi/DirectionalDiscrepancy/tree/main/Results}{https://github.com/minalatifi/DirectionalDiscrepancy/tree/main/Results}. We believe Conjecture \ref{conj:North Pole has the max Discrepancy} holds for higher $n$ as well and can be checked with the Directional Discrepancy algorithm.  However, limited computational power makes us run it only for $n \leq 125$.  Note that in the last case, the Polar Coordinates distribution puts $17234$ points on the Sphere, and the \alg algorithm uses $ 30001$ directions to cover the quarter Sphere.  The interested reader may take run for higher values of $n$ using our package at:

\href{https://github.com/minalatifi/DirectionalDiscrepancy}{https://github.com/minalatifi/DirectionalDiscrepancy}.

\subsection{Run Time}
In this section we would like to discuss the timing data of the run made to prove Conjecture \ref{conj:North Pole has the max Discrepancy}. In order to analyze the running time of our algorithm, all the relevant data has been collected from reports of running the Directional Discrepancy algorithm on Polar coordinates with parameter $n$ between 15 to 125. The run times are measured on a laptop computer with a 2.7 GHz Intel core i7 processor and 8 GB RAM.

Table \ref{tab:result} denotes the value of key variables for some of the runs as well as the run time for each phase of the \alg in minutes.  We will discuss the full data through some plots in the the rest of this section.

\begin{table}[]
    \centering
    \caption{Results of the Directional Discrepancy algorithm for $15 \leq n \leq 36$ and $108 \leq n \leq 125$.  $n$, $t$, $n_{DD}$, and $n_{CC}$ are respectively the index number for Polar Coordinates, number of points, number of directions, and number of directions need Cover Cap.} 
    \begin{tabular}[c]{|p{1.5 cm} | p{1.5 cm} | p{1.5 cm} | p{1.5 cm} | p{1.5 cm} | p{1.7 cm} | p{1.5 cm} | c}
         \hline $n$ &
         $t$ &
         $n_{DD}$ &
         $n_{CC}$ &
         phase 1 run-time (min) &
         Cover Cap run-time (min) &
         total run-time (min)
         \\ \hline15 & 250 & 3968 & 603 & 1.38 & 1.08 & 2.46
         \\ \hline16 & 284 & 3904 & 651 & 0.96 & 1.09 & 2.04
         \\ \hline17 & 318 & 4023 & 629 & 0.6 & 0.61 & 1.21
         \\ \hline18 & 359 & 3989 & 527 & 0.59 & 0.5 & 1.09
         \\ \hline19 & 400 & 3498 & 387 & 0.53 & 0.66 & 1.18
         \\ \hline20 & 441 & 4427 & 460 & 0.65 & 0.41 & 1.06
         \\ \hline21 & 486 & 4212 & 490 & 0.72 & 0.46 & 1.18
         \\ \hline22 & 538 & 4437 & 436 & 1.07 & 0.48 & 1.55
         \\ \hline23 & 586 & 5044 & 303 & 1.06 & 0.33 & 1.38
         \\ \hline24 & 634 & 4496 & 421 & 0.82 & 0.46 & 1.28
         \\ \hline25 & 690 & 5165 & 418 & 1.13 & 0.48 & 1.6
         \\ \hline26 & 749 & 4732 & 343 & 0.94 & 0.36 & 1.3
         \\ \hline27 & 807 & 4305 & 281 & 1.27 & 0.38 & 1.65
         \\ \hline28 & 864 & 5064 & 278 & 1.25 & 0.32 & 1.57
         \\ \hline29 & 928 & 5318 & 289 & 1.42 & 0.31 & 1.74
         \\ \hline30 & 994 & 5415 & 304 & 1.46 & 0.44 & 1.91
         \\ \hline31 & 1060 & 5066 & 435 & 1.64 & 0.54 & 2.18
         \\ \hline32 & 1127 & 5980 & 397 & 1.58 & 0.57 & 2.15
         \\ \hline33 & 1202 & 6194 & 421 & 1.86 & 0.56 & 2.42
         \\ \hline34 & 1277 & 6369 & 344 & 1.74 & 0.45 & 2.19
         \\ \hline35 & 1352 & 5984 & 315 & 1.98 & 0.42 & 2.41
         \\ \hline36 & 1428 & 6067 & 413 & 1.79 & 0.6 & 2.39
         \\ \hline108 & 12861 & 24578 & 1421 & 44.07 & 11.67 & 55.74
         \\ \hline109 & 13102 & 24245 & 1542 & 43.02 & 12.08 & 55.1
         \\ \hline110 & 13341 & 25006 & 1362 & 44.14 & 11.39 & 55.54
         \\ \hline111 & 13584 & 25835 & 1408 & 47.51 & 11.67 & 59.18
         \\ \hline112 & 13826 & 26206 & 1348 & 47.09 & 11.44 & 58.53
         \\ \hline113 & 14074 & 25911 & 1211 & 48.56 & 11.61 & 60.17
         \\ \hline114 & 14329 & 27120 & 1269 & 51.22 & 11.05 & 62.27
         \\ \hline115 & 14586 & 26871 & 1548 & 51.76 & 12.64 & 64.4
         \\ \hline116 & 14837 & 27795 & 1268 & 48.95 & 10.12 & 59.07
         \\ \hline117 & 15085 & 27479 & 1346 & 52.36 & 11.47 & 63.83
         \\ \hline118 & 15348 & 28819 & 1296 & 54.49 & 12.63 & 67.12
         \\ \hline119 & 15618 & 28796 & 1547 & 62.3 & 15.39 & 77.69
         \\ \hline120 & 15880 & 28223 & 1768 & 58.53 & 16.76 & 75.29
         \\ \hline121 & 16148 & 29180 & 1499 & 61.36 & 14.62 & 75.98
         \\ \hline122 & 16413 & 29258 & 1396 & 62.63 & 13.51 & 76.14
         \\ \hline123 & 16694 & 30214 & 1391 & 63.85 & 14.09 & 77.94
         \\ \hline124 & 16949 & 29746 & 1444 & 64.06 & 14.51 & 78.57
         \\ \hline125 & 17234 & 30001 & 1776 & 66.44 & 18.35 & 84.8 
         \\ \hline
    \end{tabular}
    \label{tab:result}
\end{table}

Regression is the key method that has been used in this study to find relationships between different parameters. We either fit a regression line to study the linear relation between two parameters, or plot the logarithm of both dependent and independent parameters and fit a line to it.  The latter has been used in order to capture the growth rate of the dependent variable.  In this case, if we denote the slope of the line, dependent, and independent parameters by $m$, $y$, and $x$ respectively, then we have $\log y \simeq m \log x + c$, hence $y$ grows approximately like $x^m$.

Recall that as Theorem \ref{thm:time complexity} claims, the run time of the Directional Discrepancy algorithm is $O(\frac{t}{r_{\min}^2})$. Firstly, we investigate the empirical relation between $t$ and $r_{\min}$. In Figure \ref{fig:plot_r_min_and_t}, the vertical axis denotes the values of $\frac{1}{r_{min}^2}$ and the horizontal axis shows the values of $t$ in different runs.  As Figure \ref{fig:plot_r_min_and_t} displays, there is almost a linear relation between $t$ and $\frac{1}{t_{\min}^2}$, i.e. $\frac{1}{r_{min}^2} = O(t)$. As a result, we expect the total running time to be $O(t^{2})$.  This quadratic relation can be seen in Figure \ref{fig:time_t}.  While the second order coefficient is very small in the quadratic function fitted to the data, we expect the quadratic growth be more observable for larger values of $n$.

\begin{figure}[h!]
    \centering
    \includegraphics[scale=0.5]{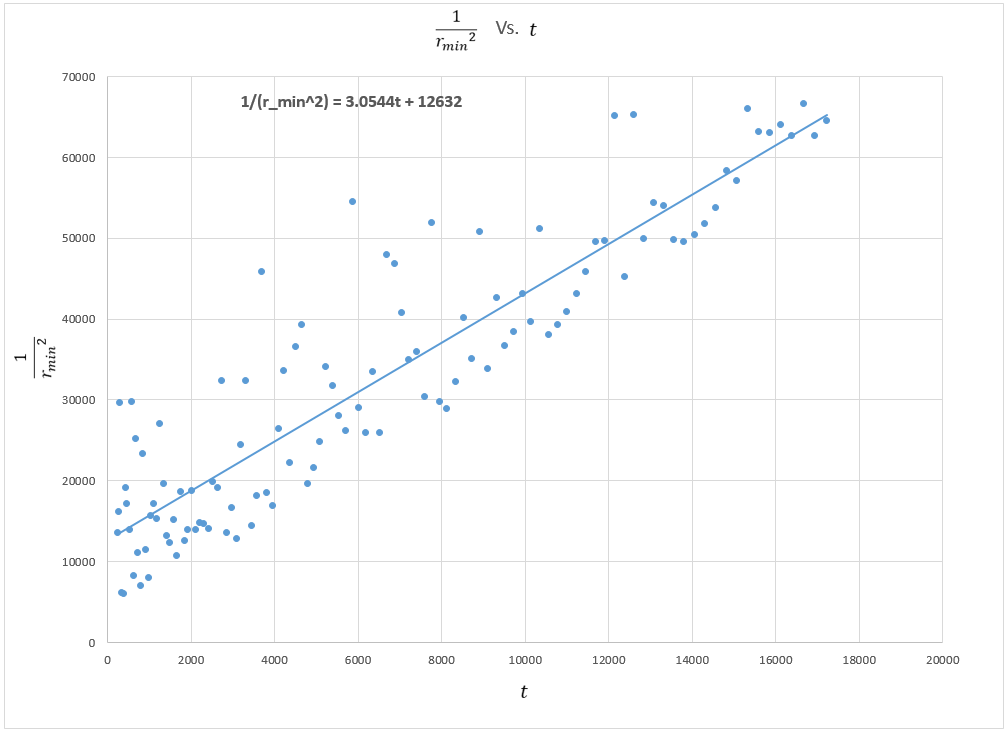}
    \caption{Relation between $r_{min}$ and $t$ for $15\leq n \leq 125$}
    \label{fig:plot_r_min_and_t}
\end{figure}

\begin{figure}[h!]
    \centering
    \includegraphics[scale=0.5]{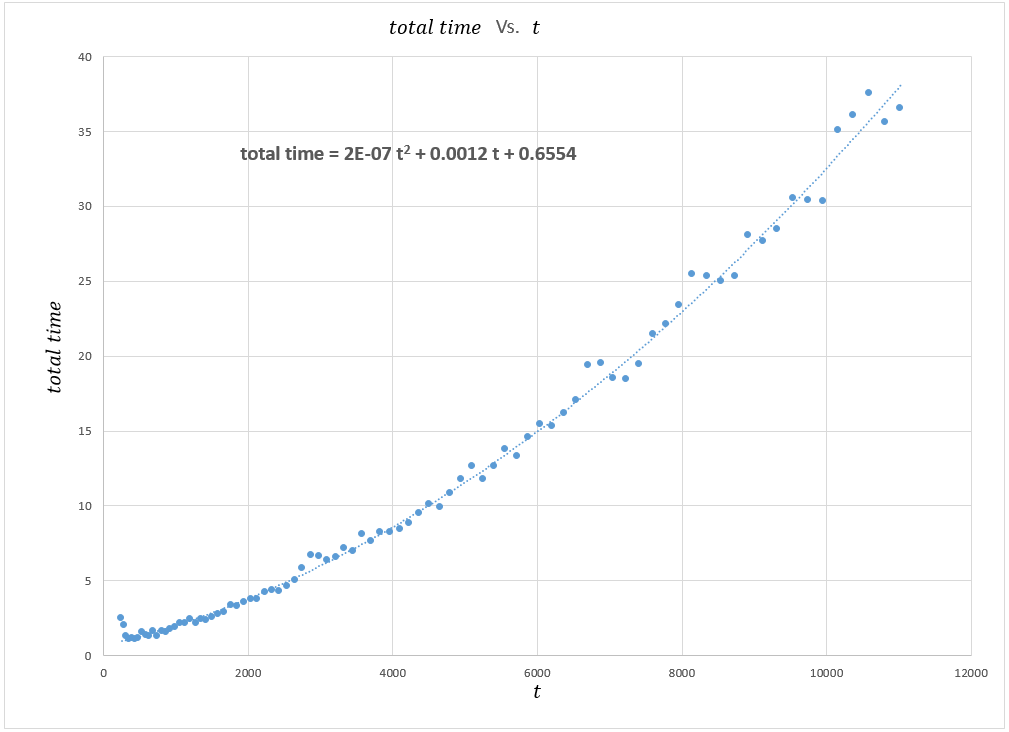}
    \caption{Total time and $t$ for $15 \leq n \leq 125$}
    \label{fig:time_t}
\end{figure}

\begin{figure}[h!]
    \centering
    \includegraphics[scale=0.5]{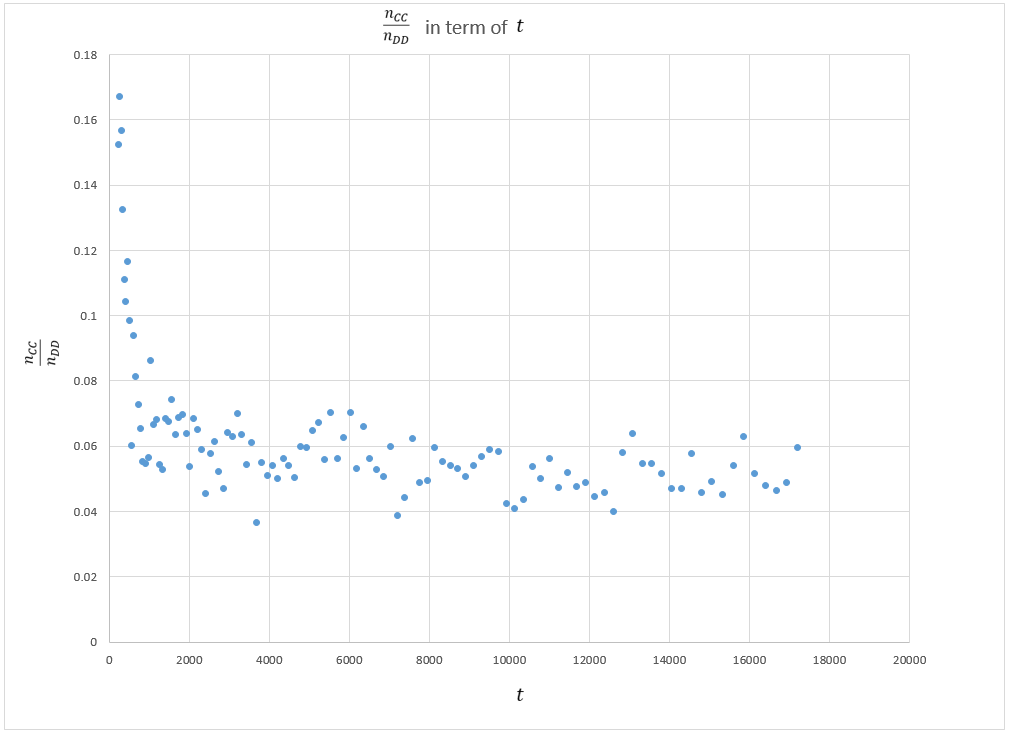}
    \caption{The ratio of $n_{CC}$ to $n_{DD}$ for $15\leq n \leq 125$}
    \label{fig:plot_n_CC_and_n_DD}
\end{figure}

Secondly, note that the total number of directions that the \alg uses to prove the upper bound, $n_{DD}$,  was a key variable in the proof of Theorem \ref{thm:time complexity}. Figure \ref{fig: n_DD_r_min} confirms that $n_{DD}$ is proportional to $\frac{1}{r_{min}^{2}}$ as claimed and proved in Theorem \ref{thm:time complexity}.  

\begin{figure}[h!]
    \centering
    \includegraphics[scale=0.5]{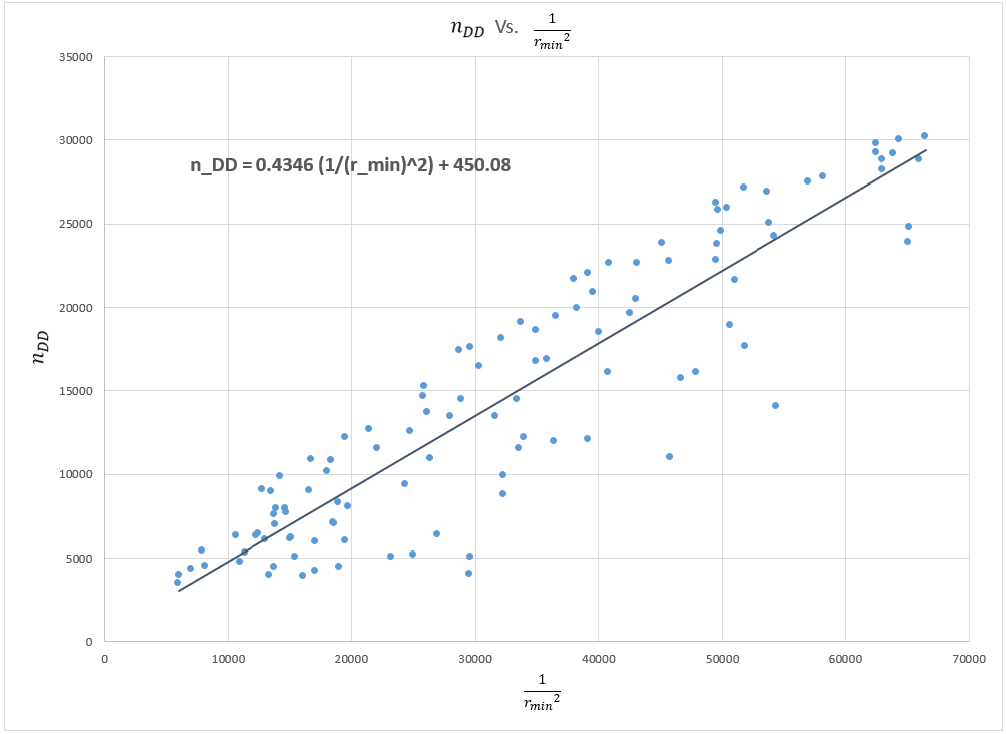}
    \caption{Linear relation between $n_{DD}$ and $\frac{1}{r_{min}^{2}}$ for $15\leq n \leq 125$}
    \label{fig: n_DD_r_min}
\end{figure} 

The time complexity of Cover Cap was discussed in Remark \ref{rem: cover cap time complexity}.  Figure \ref{fig:plot_n_CC_and_n_DD} aims to show the portion of the directions that needed to be covered by Cover Cap.  According to this Figure, for large values of $n$, almost 6 percent of directions require to be covered by Cover Cap.  Hence, the run time for Cover Cap has the same order as the run time of the phase 1, yet its value is notably smaller for each run.  Indeed, this can be observed from the raw data, reported in Table \ref{tab:result}, as well.

\begin{figure}[h!]
    \centering
    \includegraphics[scale=0.5]{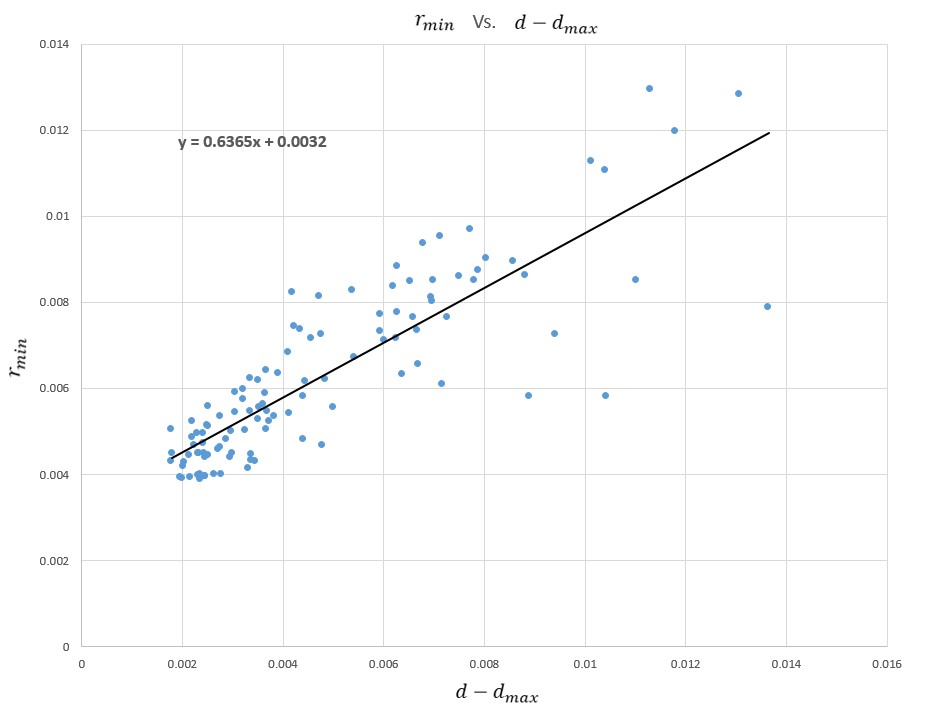}
    \caption{$r_{min}$ and $d - d_{\max}$ for $15\leq n \leq 125$}
    \label{fig:plot_r_min_d_d_max}
\end{figure}
\begin{figure}[h!]
    \centering
    \includegraphics[scale=0.5]{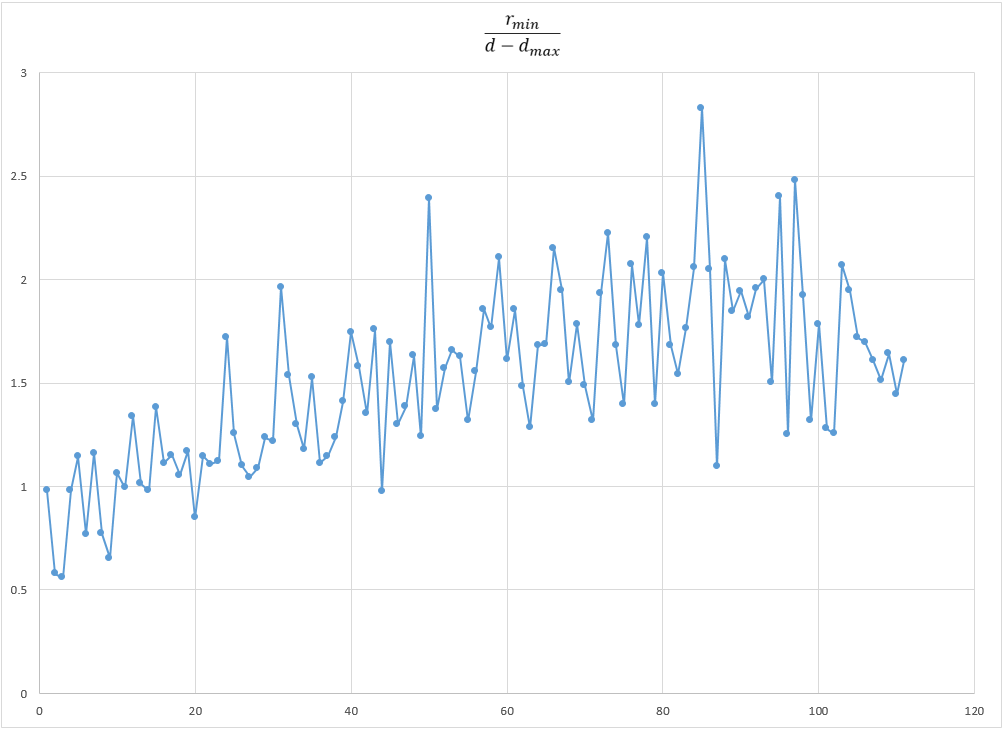}
    \caption{the ratio of $r_{\min}$ to $d - d_{\max}$ for $15\leq n \leq 125$}
    \label{fig: ratio_r_min_d_d_max}
\end{figure}

Moreover, Figure \ref{fig:plot_r_min_d_d_max} has been provided in order to support the assumption made in Remark \ref{rem:r <= 2/t}.
The positive slope of the regression line in Figure \ref{fig:plot_r_min_d_d_max} confirms that $r_{\min} = \Omega(d - d_{\max})$, which was assumed  in Remark \ref{rem:r <= 2/t} without a rigorous proof. Another helpful plot to see this has shown in Figure \ref{fig: ratio_r_min_d_d_max}.  This Figure denotes the ratio $\frac{r_{\min}}{d - d_{\max}}$ fluctuates around $1.5$ over a wide range of $n$, and hence, is not vanishing.

% \begin{figure}[h!]
%     \centering
%     \includegraphics[scale=0.5]{images/plot_time_complexity.png}
%     \caption{The relation between total run time and $t * n_{DD}$ for $80\leq n \leq 125$}
%     \label{fig:plot_total_time}
% \end{figure}

% \begin{figure}[h!]
%     \centering
%     \includegraphics[scale=0.5]{images/plot_n_DD_t.png}
%     \caption{The relation between $n_{DD}$  and $t$ for $80\leq n \leq 125$}
%     \label{fig:plot_n_DD}
% \end{figure}

% First, from a general point of view, as figure \ref{fig:total_time and T} has shown, the total time of the Directional Discrepancy algorithm is dependant to power of $T$ that is slightly larger than 1.
% \begin{figure}[h!]
%     \centering
%     \includegraphics[scale=0.5]{images/Plot_1_total time in terms of T.png}
%     \caption{Relation between Total time and $t$ for $15\leq n \leq 100$}
%     \label{fig:total_time and T}
% \end{figure}

\section{Conclusion} \label{sec: conclusion}
We introduced the notion of directed Discrepancy and proved it is upper semi continuous in some sense.  This property led us to design an approximation algorithm, Directional Discrepancy, which can offer useful information about the Discrepancy of a given distribution in a timely manner.  We showed how this tool can offer informative data about Polar Coordinates distribution and we believe it can be used, in a similar manner, to make insightful investigations for further studies.

\clearpage
\bibliography{References.bib}
\bibliographystyle{plain}

\end{document}